\documentclass[11pt]{amsart}
\usepackage{amsbsy,amssymb,amscd,amsfonts,latexsym,amstext,delarray,
amsmath,graphicx} 
\input xypic

\newtheorem{thm}{Theorem}[section]
\newtheorem{prop}[thm]{Proposition}
\newtheorem{cor}[thm]{Corollary}
\newtheorem{lem}[thm]{Lemma}
\newtheorem{guess}{Conjecture}[section]
\newtheorem{defn}[thm]{Definition}
\newtheorem{rem}[thm]{Remark}

\numberwithin{equation}{section}

\def\bG{{\mathbb G}}

\def\bL{{\mathbb L}}

\def\bP{{\mathbb P}}

\def\bT{{\mathbb T}}
\def\bU{{\mathbb U}}

\def\bZ{{\mathbb Z}}

\def\A{{\mathbb A}}
\def\C{{\mathbb C}}

\renewcommand{\P}{{\mathbb P}}
\def\Q{{\mathbb Q}}
\def\Z{{\mathbb Z}}
\def\R{{\mathbb R}}

\def\fT{{\mathfrak T}}

\def\cD{{\mathcal D}}

\def\cF{{\mathcal F}}

\def\cH{{\mathcal H}}
\def\cI{{\mathcal I}}

\def\cL{{\mathcal L}}
\def\cM{{\mathcal M}}

\def\cP{{\mathcal P}}

\def\cR{{\mathcal R}}

\def\cT{{\mathcal T}}

\def\cV{{\mathcal V}}

\def\wihat{\widehat}

\newcommand{\fm}{{\mathfrak m}}

\title{Feynman motives and deletion-contraction relations}
\author{Paolo Aluffi}
\author{Matilde Marcolli}
\address{Department of Mathematics \\ 
Florida State University \\
Tallahassee, FL 32306, USA}
\email{aluffi\@@math.fsu.edu}
\address{Department of Mathematics  \\
California Institute of Technology \\ 
Pasadena, CA 91125, USA}
\email{matilde\@@caltech.edu}

\begin{document}

\begin{abstract}
We prove a deletion-contraction formula for motivic Feynman rules given by
the classes of the affine graph hypersurface complement in the Grothendieck 
ring of varieties. 
We derive explicit recursions and generating series for these motivic 
Feynman rules under the operation of multiplying edges in a
graph and we compare it with similar formulae  for the Tutte
polynomial of graphs, both being specializations of the same universal
recursive relation. We obtain similar recursions for graphs that are
chains of polygons and for graphs obtained by replacing an edge by
a chain of triangles. We show that the deletion-contraction relation
can be lifted to the level of the category of mixed motives in the form
of a distinguished triangle, similarly to what happens in categorifications 
of graph invariants.
\end{abstract}

\maketitle

\section{Introduction}

Recently, a series of results (\cite{BroKr}, \cite{BeBro}, 
\cite{CM}, \cite{BEK}) began to reveal the existence of a 
surprising connection between the world of perturbative 
expansions and renormalization procedures in quantum field 
theory and the theory of motives and periods of algebraic 
varieties. This lead to a growing interest in investigating
algebro-geometric and motivic aspects of quantum field theory, 
see \cite{AluMa1}, \cite{AluMa2}, \cite{AluMa3},
\cite{Blo}, \cite{Brown}, \cite{CMbook}, \cite{Doryn}, \cite{Mar1},
\cite{Mar2}, \cite{Mar3}, for some recent developments.
Some of the main questions in the field revolve around the
motivic nature of projective hypersurfaces associated to
Feynman graphs. It is known by a general result of \cite{BeBro}
that these hypersurfaces generate the Grothendieck ring of
varieties, which means that, for sufficiently complicated graphs, 
they can become arbitrarily complex as motives. However, one
would like to identify explicit conditions on the Feynman graphs that 
ensure that the numbers obtained by evalating the contribution
of the corresponding Feynman integral can be described in
algebro geometric terms as periods of a sufficiently simple form,
that is, periods of mixed Tate motives. The reason to expect that this
will be the case for significant classes of Feynman graphs lies
in extensive databases of calculations of such integrals (see
\cite{BroKr}) which reveal the pervasive appearance of multiple zeta
values. 

In \cite{AluMa1}, \cite{AluMa2}, \cite{AluMa3} 
we approached the question of understanding the motivic properties of
the hypersurfaces of Feynman graphs from the point of view of
singularity theory. In fact, the graph hypersurfaces are typically
highly singular, with singularity locus of low codimension. This
has the effect that their motivic nature can often be simpler than
what one would encounter in dealing with smooth hypersurfaces. This
makes it possible to control the motivic complexity in terms of
invariants that can measure effectively how singular the hypersurfaces
are. To this purpose, in \cite{AluMa2} we looked at algebro-geometric
objects that behave like Feynman rules in quantum field theory, in the
sense that they have the right type of multiplicative behavior over disjoint
unions of graphs and the right type of decomposition relating
one-particle-irreducible (1PI) graphs and more general connected
graphs. The simplest example of such algebro-geometric Feynman rules 
is the affine hypersurface complement associated to a Feynman graph.
This behaves like the expectation value of a quantum field theory
whose edge propagator is the Tate motive $\Q(1)$. Another 
algebro-geometric Feynman rule we constructed in \cite{AluMa2} is
based on characteristic classes of singular varieties, assembled
in the form of a polynomial $C_\Gamma(T)\in \Z[T]$ associated to
a Feynman graph $\Gamma$. 

In this paper, we investigate the dependence of these
algebro-geometric Feynman rules on the underlying combinatorics of the
graphs. Our approach is based on deletion--contraction relations,
that is, formulae relating the invariant of a graph to that of
the graphs obtained by either deleting or contracting an edge.
The results we present in this paper will, in particular, answer
a question on deletion--contraction relations for Feynman 
rules asked by Michael Falk to the first author during the
Jaca conference, which motivated us to consider this problem.

It is well known that certain polynomial invariants of graphs, such as
the Tutte polynomial and various invariants obtained from it by
specializations, satisfy deletion--contraction relations. These are
akin to the skein relations for knot and link invariants, and make it
possible to compute inductively the invariant for arbitrary graphs,
by progressively reducing it to simpler graphs with fewer edges.

We first show, in \S \ref{FeyPolySec}, that the Tutte polynomial and its
specializations, the Tutte--Grothendieck invariants, define abstract
Feynman rules in the sense of \cite{AluMa2}. We observe that this
suggests possible modifications of these invariants based on applying
a Connes--Kreimer style renormalization in terms of Birkhoff
factorization. This leads to modified invariants which may be worthy
of consideration, although they lie beyond the purpose of this paper.

Having seen how the usual deletion--contraction relations of 
polynomial invariants of graphs fit in the language of Feynman rules,
we consider in \S \ref{HypersurfSec} our main object of interest, 
which is those abstract Feynman rules that are of algebro-geometric
and motivic nature, that is, that are defined in terms of the
affine graph hypersurface complement and its class in the
Grothendieck ring of varieties, or its refinement introduced in
\cite{AluMa2}, the ring of immersed conical varieties. We begin
by showing that the polynomial invariant $C_\Gamma(T)$ we
constructed in \cite{AluMa2} in terms of characteristic classes
is not a specialization of the Tutte polynomial, hence it is
likely to be a genuinely new type of graph polynomial which 
may behave in a more refined way in terms of deletion and
contraction. To obtain an explicit deletion--contraction
relation, we consider the universal motivic Feynman rule 
defined by the class $\bU(\Gamma)=[\A^n\smallsetminus \widehat X_\Gamma]$
in the Grothendieck ring of varieties of the complement of
the affine graph hypersurface of a Feynman graph. Our first main
result of the paper is Theorem \ref{delcon} where we show that
$\bU(\Gamma)$ satisfies a deletion--contraction relation of the form
$$ \bU(\Gamma) = \bL \cdot [\A^{n-1} \smallsetminus (\widehat
X_{\Gamma\smallsetminus e} \cap \widehat X_{\Gamma/e})] -
\bU(\Gamma\smallsetminus e), $$
with $\bL$ the Lefschetz motive. In \S \ref{MilnorSec} we
reinterpret this result in terms of linear systems and Milnor fibers. 
This deletion--contraction formula pinpoints rather precisely the
geometric mechanism by which non-mixed Tate motives will start to
appear when the complexity of the graph grows sufficiently. In fact, 
it is the motivic nature of the intersection of the hypersurfaces
$\widehat X_{\Gamma\smallsetminus e} \cap \widehat X_{\Gamma/e}$ that
becomes difficult to control, even when the motives of the two
hypersurfaces separately are known to be mixed Tate.

We then investigate, in \S \ref{OpSec}, certain simple operations 
on graphs, under which one can control explicitly the effect on 
the motivic Feynman rule $\bU(\Gamma)$ using the deletion--contraction
relation. The first such example is the operation that replaces an
edge in a given graph by $m$ parallel copies of the same edge. The
effect on graph hypersurfaces and their classes $\bU(\Gamma)$ of
iterations of this operation can be packaged in the form of a
generating series and a recursion, which is proved using the
deletion--contraction relation. The main feature that makes it
possible to control the whole recursive procedure in this case
is a cancellation that eliminates the class involving the intersection
of the hypersurfaces and expresses the result for arbitrary iterations
as a function of just the classes $\bU(\Gamma)$,
$\bU(\Gamma\smallsetminus e)$ and $\bU(\Gamma/e)$. Our second main
result in the paper is Theorem \ref{paraledges}, which identifies
the recursion formula and the generating function for the motivic
Feynman rules under multiplication of edges in a graph.

As a comparison, we also compute explicitly in \S \ref{multTutteSec} 
the recusion formula satisfied by the Tutte polynomial for this same 
family of operations on graphs given by multiplying edges.

Another example of operations on graphs for which the resulting
$\bU(\Gamma)$ can be controlled in terms of the deletion--contraction
relations is obtained in \S \ref{chainsSec} by looking at graphs 
that are chains of polygons. This case is further reduces in \S
\ref{lemonSec} to the case of ``lemon graphs'' given by chains of
triangles, for which an explicit recursion formula is computed.
This then gives in \S \ref{lemonadeSec} 
a similar recursion formula for the graphs obtained
from a given graph by replacing a chosen edge by a lemon graph.

We then show in \S \ref{UnivSec} that the recursion relations and
generating functions for the motivic Feynman rule and for the Tutte
polynomial under multiplication of edges in a graph are in fact 
closely related. We show that they are both specializations, for
different choice of initial conditions, of the same universal
recursion relation. We formulate a conjecture for a recursion
relation for the polynomial invariant $C_\Gamma(T)$, based on
numerical evidence collected by \cite{Stryker}. It again 
consists of a specialization of the same universal recursion relation,
for yet another choice of initial conditions. 

In the last section we show that one can think of the motive of the
hypersurface complement in the Voevodsky triangulated category of
mixed motives as a {\em categorification} of the invariant
$\bU(\Gamma)$, thinking of motives as a universal cohomology theory and
of classes in the Grothendieck ring as a universal Euler
characteristic. This categorification has properties similar to the
well known categorifications of the Jones polynomial via Khovanov
homology \cite{Khovanov} and of the chromatic polynomial and the Tutte 
polynomial  \cite{HeGuRo}, \cite{JaHeRo} via versions of graph
cohomology. In fact, in all of these cases the deletion--contraction
relations are expressed in the categorification in the form of a long
exact cohomology sequence. We show that the same happens at the
motivic level, in the form of a distinguished triangle in the
triangulated category of mixed motives.

\section{Abstract Feynman rules and polynomial invariants}\label{FeyPolySec}

We recall briefly how the Feynman rules of a perturbative
scalar field theory are defined, as a motivation for a
more general notion of abstract Feynman rule, which we
then describe. The reader who does not wish to see the
physical motivation can skip directly to the algebraic
definition of abstract Feynman rule given in Definition
\ref{abstrFeyRule}, and use that as the starting point.
Since the main results of this paper concern certain
abstract Feynman rules of combinatorial, algebro-geometric,
and motivic nature, the quantum field theoretic notions
we recall here serve only as background and motivation.

\subsection{Feynman rules in perturbative quantum field theory}\label{QFTsec}

In perturbative quantum field theory, the evaluation of
functional integrals computing expectation values of
physical observables is obtained by expanding the
integral in a perturbative series, whose terms are
labeled by graphs, the Feynman graphs of the theory,
whose valences at vertices are determined by the
Lagrangian of the given physical theory. The number
of loops of the Feynman graphs determines how far
one is going into the perturbative series in order
to evaluate radiative corrections to the expectation
value. The contribution of individual graphs to the
perturbative series is determined by the {\em Feynman rules}
of the given quantum field theory. In the case of a
scalar theory, these can be summarized as follows.

A graph $\Gamma$ is a
Feynman graph of the theory if all vertices have 
valence equal to the degree of one of the monomials
in the Lagrangian of the theory. Feynman graphs have
internal edges, which are thought of as matching pairs
of half edges connecting two of the vertices of the graph, 
and external edges, which are unmatched half edges 
connected to a single vertex. A graph is
1-particle-irreducible (1PI) if it 
cannot be disconnected by removal of a single (internal) edge.

We consider a scalar quantum field theory specified
by a Lagrangian of the form
\begin{equation}\label{Lagr}
\cL(\phi) = \int \left( \frac{1}{2} (\partial_\mu \phi)^2 
+ \frac{m^2}{2}\phi^2 + \cP(\phi) \right) dv ,
\end{equation}
where we use Euclidean signature in the metric on the underlying spacetime $\R^D$ 
and the interaction term is a polynomial of the form
\begin{equation}\label{Pphi}
 \cP(\phi) = \sum_k \frac{\lambda_k}{k!} \phi^k .
\end{equation}
In the following we will treat the dimension $D$ of the underlying
spacetime as a variable parameter.

To a connected Feynman graph of a given scalar quantum
field theory one assigns a function $V(\Gamma,p_1,\ldots,p_N)$
of the external momenta in the following way.

Each internal edge $e\in E_{int}(\Gamma)$ contributes a momentum 
variable $k_e\in \R^D$ and the function of the external momenta is
obtained by integrating a certain density function over the momentum
variables of the internal edges,
\begin{equation}\label{VGammapi}
V(\Gamma,p_1,\ldots,p_N)= \int \cI_\Gamma(p_1,\ldots,p_N,k_1,\ldots, k_n) \frac{d^D k_1}{(2\pi)^D}
\cdots \frac{d^D k_n}{(2\pi)^D},
\end{equation}
for $n=\# E_{int}(\Gamma)$. We write $k=(k_e)$ for the collection of
all the momentum variables assigned to the internal edges.

The term $\cI_\Gamma(p_1,\ldots,p_N,k_1,\ldots, k_n)$ is constructed
according to the following procedure.
Each vertex $v\in V(\Gamma)$ contributes a factor of $\lambda_v (2\pi)^D$, where
$\lambda_v$ is the coupling constant of the monomial 
in the interaction term \eqref{Pphi} in the Lagrangian of order equal
to the valence of $v$. One also imposes  
a conservation law on the momenta that flow through a 
vertex, 
\begin{equation}\label{deltav}
\delta_v(k) := \delta(\sum_{s(e)=v} k_e-\sum_{t(e)=v} k_e),
\end{equation}
written after chosing an orientation of the edges of the graph, so
that $s(e)$ and $t(e)$ are the source and target of an edge $e$. When
a vertex is attached to both internal and external edges, the
conservation law \eqref{deltav} at that vertex will be of an analogous
form $\delta_v(k,p)$, involving both the
$k$ variables of the momenta along internal edges and the $p$ variables
of the external momenta. We will
see later that the dependence on the choice of the orientation disappears
in the final form of the Feynman integral.

Each internal edge $e\in E_{int}(\Gamma)$ contributes an inverse propagator, that is, 
a term of the form $q_e^{-1}$, where $q_e$ is a quadratic form, which in the
case of a scalar field in the Euclidean signature is given by 
\begin{equation}\label{propagatorEu}
q_e(k_e)=k_e^2 + m^2.
\end{equation}

Each external edge $e\in E_{ext}(\Gamma)$ contributes an inverse 
propagator $q_e(p_e)^{-1}$,
with $q_e(p_e)=p_e^2 + m^2$. The external momenta are assigned so that they satisfy
the conservation law $\sum_e p_e=0$, when summed over the oriented external edges.

The integrand $\cI_\Gamma(p_1,\ldots,p_N,k_1,\ldots, k_n)$ is then a product
\begin{equation}\label{Iprodrules}
\prod_{v\in V(\Gamma)} \lambda_v (2\pi)^D \, \, \delta_v(k,p) \,\,
\prod_{e\in E_{int}(\Gamma)} q_e(k_e)^{-1}  
\prod_{e\in E_{ext}(\Gamma)} q_e(p_e)^{-1}.
\end{equation}

The Feynman rules defined in this way satisfy two main properties, which
follow easily from the construction described above (see \cite{LeBellac},
\cite{Mar3}).

The Feynman rules are multiplicative over disjoint unions of
graphs (hence one can reduce to considering only connected graphs):
\begin{equation}\label{VGammaUnion}
V(\Gamma,p_1,\ldots,p_{N_1},p_1',\ldots,p_{N_2}')=V(\Gamma_1, p_1,\ldots, p_{N_1}) 
V(\Gamma_2, p_1',\ldots,p_{N_2}'),
\end{equation}
for a disjoint union $\Gamma = \Gamma_1 \cup \Gamma_2$, of two Feynman 
graphs $\Gamma_1$ and $\Gamma_2$, with external momenta
$p=(p_1,\ldots,p_{N_1})$ and $p'=(p_1',\ldots,p_{N_2}')$, respectively.

Any connected graph $\Gamma$ can be obtained from a finite tree $T$
by replacing vertices $v$ of $T$ with 1PI graphs $\Gamma_v$ with number
of external edges equal to the valence of the vertex $v$. Then the
Feynman rules satisfy
\begin{equation}\label{FeyNot1PI}
V(\Gamma,p)= \prod_{v\in V(T), e\in E_{int}(T), v\in\partial(e)} 
V(\Gamma_v,p_v)\,\, q_e(p_v)^{-1}\,\, 
\delta((p_v)_e-(p_{v'})_e).
\end{equation}
The delta function in this expression matches the external momenta
of the 1PI graphs inserted at vertices sharing a common edge.

Up to a factor containing the inverse propagators of the external
edges and the coupling constants of the vertices, we write
$$ V(\Gamma,p_1,\ldots,p_N)= C\varepsilon(p_1,\ldots,p_N)\, 
U(\Gamma,p_1,\ldots,p_N),$$
with $C=\prod_{v\in V(\Gamma)} \lambda_v (2\pi)^D$ and
$\varepsilon(p_1,\ldots,p_N) =\prod_{e\in E_{ext}(\Gamma)}
q_e(p_e)^{-1}$ and the remaining term is
\begin{equation}\label{Uint}
U(\Gamma,p_1,\ldots,p_N) = \frac{\delta(\sum_{i=1}^n \epsilon_{v,i} k_i 
+ \sum_{j=1}^N \epsilon_{v,j} p_j)}{ q_1(k_1)\cdots q_n(k_n) } ,
\end{equation}
where we have written the delta functions $\delta_v(k,p)$ of
\eqref{deltav} equivalently in terms of the edge-vertex 
incidence matrix of the graph, 
$\epsilon_{v,i}=\pm 1$ when $v=t(e)$ or $v=s(e)$ and
$\epsilon_{v,i}=0$ otherwise. The Feynman integrals \eqref{Uint}
still satisfy the two properties \eqref{VGammaUnion} and
\eqref{FeyNot1PI}. 

Notice that the property \eqref{FeyNot1PI} expressing the Feynman
rule for connected graphs in terms of Feynman rules for 1PI graphs
has a simpler form in the case where either all external momenta
are set equal to zero and the theory is massive ($m\neq 0$), or
all external momenta are equal. In such
cases \eqref{FeyNot1PI} reduces to a product
$$ U(\Gamma,p)=U(L)^{\# E_{int}(T)} \prod_{v\in V(T)} U(\Gamma_v,p_v), $$
with $U(L)$ the inverse propagator assigned to a single edge.

\subsection{Abstract Feynman rules}\label{AbsFRsec}

In \cite{AluMa2} we abstracted the two properties of Feynman rules recalled
above and used them to define a class of {\em algebro-geometric} Feynman rules.

More precisely, we defined an abstract Feynman rule in the following
way.

\begin{defn}\label{abstrFeyRule}
An abstract Feynman rule is a map
from the set of (isomorphism classes of) 
finite graphs to a commutative ring $\cR$,
with the property that it is multiplicative over disjoint unions of
graphs,
\begin{equation}\label{Umultipl}
U(\Gamma_1 \cup \Gamma_2) = U(\Gamma_1) U(\Gamma_2),
\end{equation}
and such that, for a connected graph $\Gamma = \cup_{v\in V(T)} \Gamma_v$
obtained by inserting 1PI graphs $\Gamma_v$ at the vertices of a tree $T$,
it satisfies
\begin{equation}\label{U1PIT}
U(\Gamma) = U(L)^{\# E_{int}(T)} \prod_{v\in V(T)} U(\Gamma_v), 
\end{equation} 
where $U(L)$ is the inverse propagator, that is, the value assigned 
to the graph consiting of a single edge.
\end{defn}

The multiplicative property with respect to disjoint
unions of graphs, together with the second property
which implies that 1PI graphs are sufficient to
determine completely the Feynmal rule, means that
an abstract Feynman rule with values in $\cR$ can be
reformulated as a ring homomorphism from a
Hopf algebra $\cH$ of Feynman graphs to $\cR$. In this general
setting, since we are not choosing a particular Lagrangian
of the theory, the Hopf algebra is not the usual Connes--Kreimer
Hopf algebra \cite{CoKr}, which depends on the Lagrangian of a particular
theory but the larger Hopf algebra referred to in \cite{BK}
and \cite{Kr} as the ``core Hopf algebra''. As an algebra
(or a ring) this is a polynomial algebra generated by all
1PI graphs and the coproduct is of the form
\begin{equation}\label{coprod}
\Delta(\Gamma) =\Gamma \otimes 1 + 1 \otimes \Gamma + 
\sum_{\gamma \subset \Gamma} \gamma \otimes \Gamma/\gamma,
\end{equation}
where the sum is over subgraphs whose connected components
are 1PI. The quotient $\Gamma/\gamma$ is obtained by shrinking
each component of $\gamma$ to a single vertex. The Hopf algebra
is graded by loop number (or by number of internal edges) and
the antipode is defined inductively by
$$ S(\Gamma)=-\Gamma - \sum S(\gamma)\, \Gamma/\gamma. $$ 
Notice that the multiplicative property of Feynman rules
only relates to the algebra, not the coalgebra, structure
of the Hopf algebra of Feynman graphs. Where the coproduct
and the antipode enter essentially is in the renormalization
of abstract Feynman rules, which can again be formulated
purely algebraically in terms of a Rota--Baxter structure 
of weight $-1$ on the target ring $\cR$ (see \cite{CoKr2}, \cite{EFGK}).

\medskip

We now show that the notion of abstract Feynman rule is very
natural. In fact, a broad range of classical combinatorial invariants 
of graphs define abstract Feynman rules. 

\subsection{Tutte--Grothendieck polynomials as a Feynman rules}\label{TutteSec}

For a finite graph $\Gamma$, one denotes by $\Gamma \smallsetminus e$
the graph obtained by deleting an edge $e \in E(\Gamma)$ and by
$\Gamma/e$ the graph obtained by contracting an edge $e\in E(\Gamma)$
to a vertex. They are called, respectively, the deletion and
contraction of $\Gamma$ at $e$. A class of invariants of graphs that
behave well with respect to the operations of deletion and contraction
are known as the Tutte--Grothendieck invariants \cite{Bollobas},
\cite{Bryla}, \cite{Crapo}. The terminology comes from the Tutte
polynomial, which is the prototype example of such invariants,
and from the formulation in terms of Grothendieck rings of certain
categories, as in \cite{Bryla}.

Tutte--Grothendieck invariants of graphs are defined as functions
$F(\Gamma)$ from the set of (isomorphism classes of) finite graphs to a
polynomial ring $\C[\alpha,\beta,\gamma,x,y]$ which satisfy the
following properties.
\begin{itemize}
\item $F(\Gamma)=\gamma^{\# V(\Gamma)}$ if the set of edges is 
empty, $E(\Gamma)=\emptyset$.
\item $F(\Gamma)=x F(\Gamma \smallsetminus e)$ if the edge $e \in E(\Gamma)$ is a
bridge.
\item $F(\Gamma) = y F(\Gamma/e)$ if 
$e \in E(\Gamma)$ is a looping edge.
\item For $e\in E(\Gamma)$ not a bridge nor a looping edge,
\begin{equation}\label{delcontr}
F(\Gamma)=\alpha F(\Gamma/e) + \beta F(\Gamma \smallsetminus e).
\end{equation}
\end{itemize}

Recall that an edge is a bridge (or isthmus) if the removal
of $e$ disconnects the graph $\Gamma$. A looping edge is an
edge that starts and ends at the same vertex. The relation
\eqref{delcontr} is the deletion--contraction relation. By
repeatedly applying it until one falls into one of the
other case, this makes it possible to completely determine
the value of a Tutte--Grothendieck invariant for all graphs.

Tutte--Grothendieck invariants are specializations of the
Tutte polynomial. The latter is defined by the properties
that 
\begin{equation}\label{Tutte1}
\cT_\Gamma (x,y)= x^i y^j,
\end{equation}
if the graph $\Gamma$ consists of $i$ bridges, $j$ looping edges
and no other edges, and by the deletion--contraction relation
\begin{equation}\label{Tutte2}
\cT_{\Gamma}(x,y)= \cT_{\Gamma\smallsetminus e}(x,y) + \cT_{\Gamma/e}(x,y).
\end{equation}
Clearly the relation \eqref{Tutte2} together with \eqref{Tutte1}
determine the Tutte polynomial for all graphs. The closed
formula is given by the ``sum over states'' formula
\begin{equation}\label{Tuttepoly}
\cT_{\Gamma}(x,y)= \sum_{\gamma \subset \Gamma} (x-1)^{\#
V(\Gamma)-b_0(\Gamma)-(\# V(\gamma)-b_0(\gamma))}
(y-1)^{\#E(\gamma)-\# V(\gamma) +b_0(\gamma)},
\end{equation}
where the sum is over subgraphs $\gamma \subset \Gamma$ with
vertex set $V(\gamma)=V(\Gamma)$ and edge set 
$E(\gamma)\subset E(\Gamma)$. This can be written equivalently as
$$ \cT_{\Gamma}(x,y)= \sum_{\gamma \subset \Gamma} (x-1)^{b_0(\gamma)-b_0(\Gamma)}
(y-1)^{b_1(\gamma)}. $$
An equivalent way to express the recursive relations computing
the Tutte polynomials is the following:
\begin{itemize}
\item If $e\in E(\Gamma)$ is neither a looping edge nor a bridge
the deletion--contraction relation \eqref{Tutte2} holds.
\item If $e\in E(\Gamma)$ is a looping edge then
$$ \cT_\Gamma(x,y)= y \cT_{\Gamma/e}(x,y). $$
\item If $e\in E(\Gamma)$ is a bridge then
$$ \cT_\Gamma(x,y)= x \cT_{\Gamma \smallsetminus e}(x,y) $$
\item If $\Gamma$ has no edges then $T_\Gamma(x,y)=1$.
\end{itemize}

A Tutte--Grothendieck invariant satisfying \eqref{delcontr} is then
obtained from the Tutte polynomial by specialization
\begin{equation}\label{Fspecialize}
 F(\Gamma)= \gamma^{b_0(\Gamma)} \alpha^{\# V(\Gamma)-b_0(\Gamma)} 
\beta^{b_1(\Gamma)} \, \cT_\Gamma
(\frac{\gamma x}{\alpha}, \frac{y}{\beta}). 
\end{equation}

Among the invariants that can be obtained as 
specializations of the Tutte polynomial are the chromatic polynomial of graphs
and the Jones polynomial of links, viewed as an invariant of
an associated planar graph, \cite{Thist}, \cite{Jaeger}. 

The chromatic polynomial $P(\Gamma,\lambda)$ is a specialization of the
Tutte polynomial through
\begin{equation}\label{chromatic}
P(\Gamma,\lambda)= (-1)^{\# V(\Gamma) - b_0(\Gamma)} \lambda^{b_0(\Gamma)}
\cT_\Gamma(1-\lambda,0). 
\end{equation}
In the case of an alternating link $L$, the Jones polynomial is a
specialization of the associated (positive) checkerboard graph
$\Gamma_+$ by
$$ J(L,t)= (-1)^w t^{(\# V(\Gamma_-)-\# V(\Gamma_+)+3w)/4} \,\,
\cT_{\Gamma_+}(-t,-1/t), $$
with $w$ the writhe (algebraic crossing number) and $\Gamma_\pm$ the
positive and negative checkerboard graphs associated to $L$,
\cite{Bollobas}. 

\begin{prop}\label{TutteFeynman}
The Tutte polynomial invariant defines an abstract Feynman rule
with values in the polynomial ring $\C[x,y]$, by assigning
\begin{equation}\label{UTutte}
U(\Gamma) = \cT_\Gamma(x,y), \ \ \ \text{ with inverse propagator } \ \
\  U(L)=x. 
\end{equation}
Similarly, any Tutte--Grothendieck invariant determines
an abstract Feynman rule with values in $\C[\alpha,\beta,\gamma,x,y]$
by assigning $U(\Gamma)=F(\Gamma)$ with inverse propagator $U(L)=x$.
\end{prop}

\proof It suffices to check that the Tutte polynomial is
multiplicative over disjoint unions of graphs and that,
under the decomposition of connected graphs into a tree 
with 1PI graphs inserted at the vertices, it satisfies the
property \eqref{U1PIT}. The multiplicative property is
clear from the closed expression \eqref{Tuttepoly}, since
for $\Gamma = \Gamma_1 \cup \Gamma_2$ we can identify subgraphs
$\gamma \subset \Gamma$ with $V(\gamma)=V(\Gamma)$ and
$E(\gamma)\subset E(\Gamma)$ with all possible pairs of
subgraphs $(\gamma_1,\gamma_2)$ with $V(\gamma_i)=V(\Gamma_i)$
and $E(\gamma_i)\subset E(\Gamma_i)$, with $b_0(\gamma)=
b_0(\gamma_1)+b_0(\gamma_2)$, $\# V(\Gamma)=\# V(\Gamma_1) + \#
V(\Gamma_2)$, and $\# E(\gamma)=\# E (\gamma_1)+\# E(\gamma_2)$.
Thus, we get
$$ \begin{array}{rl}
\cT_\Gamma(x,y)= & \sum_{\gamma=(\gamma_1,\gamma_2)} 
(x-1)^{b_0(\gamma_1)+b_0(\gamma_2)-b_0(\Gamma)}
(y-1)^{b_1(\gamma_1)+b_1(\gamma_2)} \\[2mm] = &  \cT_{\Gamma_1}(x,y) \,
\cT_{\Gamma_2}(x,y). \end{array} $$
The second property for connected and 1PI graphs follows 
from the fact that, when writing a connected graph in the
form $\Gamma = \cup_{v\in V(T)} \Gamma_v$, with $\Gamma_v$ 1PI
graphs inserted at the vertices of the tree $T$, the internal edges of
the tree are all bridges in the resulting graph, hence the
property of the Tutte polynomial for the removal of bridges
gives
$$ \cT_\Gamma(x,y)= x^{\# E_{int}(T)}\,\, \cT_{\Gamma\smallsetminus \cup_{e\in
E_{int}(T)} e}(x,y). $$
Then one obtains an abstract Feynman rule with values in $\cR=\C[x,y]$
of the form \eqref{UTutte}.

In fact, the multiplicative property follows from the same
property for the Tutte polynomial and the specialization
formula \eqref{Fspecialize}. The case of connected and 1PI
graphs again follow from the property of Tutte--Grothendieck
invariants that $F(\Gamma)=x F(\Gamma \smallsetminus e)$,
when $e\in E(\Gamma)$ is a bridge, exactly in the same 
way as in the case of $\cT_\Gamma(x,y)$.
\endproof

This implies that both the chromatic and the Jones polynomial,
for instance, can be regarded as abstract Feynman rules.

As we have mentioned above, whenever the ring $\cR$ where
an abstract Feynman rule takes values has the structure
of a Rota--Baxter algebra of weight $-1$, the Feynman rule
can be renormalized. This works in the following way.

A Rota--Baxter operator of weight $\lambda$ is a linear 
operator $\fT: \cR \to \cR$ satisfying 
\begin{equation}\label{RBop}
\fT(x) \fT(y)= \fT(x\fT(y)) +\fT(\fT(x)y) + \lambda \fT(xy).
\end{equation} 
In the case where $\lambda=-1$, such an operator determines
a decomposition of the ring $\cR$ into two commutative
unital rings $\cR_\pm$ defined by $\cR_+=(1-\fT)\cR$ and
$\cR_-$ the ring obtained by adjoining a unit to the
nonunital $\fT \cR$. 
An example of $(\cR,\fT)$ is given by
Laurent series with the projection onto the polar part. 

The Connes--Kreimer interpretation \cite{CoKr} \cite{CoKr2} 
of the BPHZ renormalization procedure as a Birkhoff factorization
of loops with values in the affine group scheme dual to the
Hopf algebra of Feynman graphs can be formulated equivalently
in terms of the Rota--Baxter structure \cite{EFGK}. The 
Connes--Kreimer recursive formula for the Birkhoff factorization
of an algebra homomorphism $U: \cH \to \cR$ is given as in \cite{CoKr} by 
\begin{equation}\label{Birkhoff}
\begin{array}{rl}
U_-(\Gamma) = & - \displaystyle{ \fT \left( U(\Gamma) + \sum_{\gamma \subset \Gamma}
U_-(\gamma) U(\Gamma/\gamma) \right) } \\[3mm]
U_+(\Gamma) = & \displaystyle{ (1-\fT) \left( U(\Gamma) + \sum_{\gamma \subset \Gamma}
U_-(\gamma) U(\Gamma/\gamma) \right)} .
\end{array}
\end{equation}
Of these, the $U_-$ term gives the counterterms and the $U_+$ gives
the renormalized value.

Thus, for example, when one considers specializations of the Tutte
polynomial such as the Jones polynomial, which take values in a
ring of Laurent series, one can introduce a renormalized version
of the invariant obtained by performing the Birkhoff factorization
of the corresponding character of the Hopf algebra of Feynman graphs. 
It would be interesting to see if properties of the coefficients
of the Jones polynomial, such as the fact that they are not finite
type invariants, may be affected by this renormalization procedure.

\section{Graph hypersurfaces and deletion-contraction relations}\label{HypersurfSec}

In \cite{AluMa2} we considered in particular abstract Feynman rules
that are algebro-geometric or motivic, which means that they factor
through the information on the affine hypersurface defined by the
Kirchhoff polynomial of the graph, which appears in the parametric
form of Feynman integrals. We recall here how the graph hypersurfaces
are defined and how they arise in the original context of parametric
Feynman integrals. We recall from \cite{AluMa2} how one can use the
affine hypersurface complement to define algebro-geometric and 
motivic Feynman rules, and we then prove that motivic Feynman rules
satisfy a more complicated variant of the deletion--contraction
relation discussed above.

\subsection{Parametric Feynman integrals and graph hypersurfaces}\label{ParFIsec}

The Feynman rules \eqref{Uint} for a
scalar quantum field theory can be reformulated in terms of
Feynman parameters (see \cite{BjDr2}, \cite{ItZu}) in the form
of an integral of an algebraic differential form on a cycle
with boundary in the complement of a hypersurface defined by
the vanishing of the graph polynomial. The parametric form
of the Feynman integral, in the massless case $m=0$, is given by 
\begin{equation}\label{UGammaPPsi}
U(\Gamma,p_1,\ldots,p_N) =
\frac{\Gamma(n-\frac{D\ell}{2})}{(4\pi)^{D\ell/2}} 
\int_{\sigma_n}
\frac{P_\Gamma(t,p)^{-n+D\ell/2}\omega_n}{\Psi_\Gamma(t)^{-n +D(\ell+1)/2}} ,
\end{equation}
where $n=\# E_{int}(\Gamma)$ and $\ell = b_1(\Gamma)$. The domain
of integration is the simplex $\sigma_n = \{ t \in \R_+^n\,|\, \sum_i
t_i =1 \}$.
The Kirchhoff--Symanzik polynomial $\Psi_\Gamma(t)$ is given by 
\begin{equation}\label{SpanTrees}
\Psi_\Gamma(t) = \sum_{T\subset \Gamma} \prod_{e\notin E(T)} t_e,
\end{equation}
where the sum is over all the spanning forests $T$ (spanning trees
in the connected case) of the graph $\Gamma$ and for
each spanning forest the product is over all edges of $\Gamma$ that are not in that
spanning forest.
The polynomial $P_\Gamma(t,p)$, often referred to as the second
Symanzik polynomial, is similarly defined in terms of
the combinatorics of the graph, using cut sets instead of spanning
trees, and it depends explicitly on the external momenta of the
graph (see \cite{BjDr2} \S 18). 
In the following, we assume for simplicity to work in the
``stable range'' where $-n+D\ell/2 >0$. In this case, further 
assuming that for general choice of the external momenta the
polynomials $\Psi_\Gamma(t)$ and $P_\Gamma(t,p)$ do not have
common factors, the parametric Feynman integral \eqref{UGammaPPsi}
is defined in the complement of the hypersurface
\begin{equation}\label{hatXGamma}
\wihat X_\Gamma = \{ t \in \A^n \,|\, \Psi_\Gamma(t) =0 \}.
\end{equation}
Since $\Psi_\Gamma$ is homogeneous of degree $\ell$, one can
reformulate the period computation in projective space in terms
of the hypersurface $$X_\Gamma =\{ t \in \P^{n-1}\,|\,
\Psi_\Gamma(t)=0\},$$ see \cite{BEK}. 
Up to a divergent Gamma-factor, one is then interested in understanding the 
nature of the remaining integral (the residue of the Feynman graph)
\begin{equation}\label{UGammaPPsiNoGamma}
 \int_{\sigma_n}
\frac{P_\Gamma(t,p)^{-n+D\ell/2}\omega_n}{\Psi_\Gamma(t)^{-n +D(\ell+1)/2}} .
\end{equation}
viewed (possibly after eliminating divergences) as a period of an
algebraic variety. The complexity of the period depends on the motivic
complexity of the part of the cohomology of the algebraic variety that
is involved in the period evaluation. In this case, the integration is
on the domain $\sigma_n$ with boundary $\partial \sigma_n \subset
\wihat\Sigma_n$ contained in the divisor $\wihat\Sigma_n \subset \A^n$ 
given by the union of coordinate hyperplanes $\wihat\Sigma_n =\{ t\in
\A^n\,|\, \prod_i t_i =0 \}$, hence the relative cohomology involved
in the period computation is 
\begin{equation}\label{relcohom}
H^{n-1}(\P^{n-1}\smallsetminus X_\Gamma, \Sigma_n \smallsetminus
(X_\Gamma \cap \Sigma_n)),
\end{equation}
where $\Sigma_n$ is the divisor of coordinate hyperplanes in $\P^{n-1}$.
The main question one would like to address is under what conditions
on the graph this cohomology is a realization of a mixed Tate motive,
which in turn gives a strong bound on the complexity of the periods.

A way to understand the motivic complexity of \eqref{relcohom} is to
look at classes in the Grothendieck ring of varieties. A result of
Belkale--Brosnan \cite{BeBro} shows that the classes $[X_\Gamma]$ of
the graph hypersurfaces generate the Grothendieck ring, so they can
be arbitrarily complex as motives and not only of mixed Tate type.
It is still possible, however, that the piece of the cohomology 
involved in \eqref{UGammaPPsiNoGamma} may still be mixed Tate even if
$X_\Gamma$ itself contains non-mixed Tate strata. 

\subsection{Algebro-geometric and motivic Feynman rules}\label{MotFRsec}

Coming back to abstract Feynman rules, we observed in \cite{AluMa2}
that the affine hypersurface complement $\A^n \smallsetminus 
\wihat X_\Gamma$ behaves like a Feynman rule, in the sense that it
satisfies the multiplicative property under disjoint unions of graphs
\begin{equation}\label{prodAffCompl}
\A^n \smallsetminus \wihat X_\Gamma = (\A^{n_1} \smallsetminus \wihat
X_{\Gamma_1}) \times (\A^{n_2} \smallsetminus \wihat X_{\Gamma_2}),
\end{equation}
for $\Gamma = \Gamma_1 \cup \Gamma_2$ a disjoint union. The
role of the inverse propagator is played by the affine line $\A^1$.

We introduced in \cite{AluMa2} a {\em Grothendieck ring of immersed
conical varieties} $\cF$ which is generated by the equivalence classes
$[\wihat X ]$ up to linear changes of coordinates of varieties $\wihat X
\subset \A^N$ embedded in some affine space, that are defined by 
homogeneous ideals (affine cones over projective varieties), with
the usual inclusion--exclusion relation
$$ [\wihat X]= [\wihat Y] + [\wihat X\smallsetminus \wihat Y] $$
for $\wihat Y \subset \wihat X$ a closed embedding. This maps to the
usual Grothendieck ring of varieties $K_0(\cV)$ by passing to isomorphisms
classes of varieties. 

We then defined in \cite{AluMa2} algebro-geometric and motivic 
Feynman rules in the following way.

\begin{defn}\label{motFeyRules}
An algebro geometric Feynman rule is an abstract Feynman rule
$U: \cH \to \cR$, which factors through the Grothendieck ring 
of immersed conical varieties, 
\begin{equation}\label{UAffcompl}
U(\Gamma) = I([\A^n\smallsetminus \wihat X_\Gamma]),
\end{equation}
where $[\A^n\smallsetminus \wihat X_\Gamma]$ is the class in $\cF$
of the affine graph hypersurface complement and $I: \cF \to \cR$ 
is a ring homomorphism.
A motivic Feynman rule is an abstract Feynman rule that similarly
factors through the Grothendieck ring of varieties as in
\eqref{UAffcompl} with $[\A^n\smallsetminus \wihat X_\Gamma]$ the class
in $K_0(\cV)$ and a ring homomorphism $I: K_0(\cV)\to \cR$.
\end{defn}

It is then natural to ask whether these abstract Feynman rules, like
the examples of abstract Feynman rules we have described in \S
\ref{TutteSec} above, satisfy deletion--contraction relations.
We show in \S \ref{delconMotSec} below that there is a 
deletion--contraction relation for
the graph hypersurfaces and their classes in the Grothendieck ring,
which is, however, of a more subtle form than the one satisfied by
Tutte--Grothendieck invariants. We first show that the polynomial
invariant of graphs we introduced in \cite{AluMa2} as an example
of an algebro-geometric Feynman rule which is not motivic (it
does not factor through the Grothendieck ring) is {\em not} 
a specialization of the Tutte polynomial. 

\subsection{The Chern--Schwartz--MacPherson Feynman rule}\label{CSMsec}

In particular, we constructed in \cite{AluMa2} an algebro-geometric 
Feynman rule given by a polynomial invariant $C_\Gamma(T)$ constructed
using Chern--Schwartz--MacPherson characteristic classes of singular
varieties. Without going into the details of the definition and
properties of this invariant, for which we refer the reader to 
\cite{AluMa2}, we just mention briefly how it is obtained. One
obtains a ring homomorphism $I_{CSM}: \cF \to \Z[T]$ from the
ring of immersed conical varieties to a polynomial ring by 
assigning to the class $[\wihat X]$ of a variety in $\cF$ the polynomial
$$ I_{CSM}([\wihat X])= a_0 + a_1 T + \cdots a_N T^N $$
where $\wihat X \subset \A^N$ (viewed as a locally closed subscheme of $\P^N$) 
has Chern--Schwartz--MacPherson (CSM) class
$$ c_*(1_{\wihat X})=a_0[\P^0] + a_1 [\P^1] + \cdots a_N [\P^N] $$
in the Chow group (or homology) of $\P^N$. It is shown in
\cite{AluMa2} that this is well defined and is indeed a ring
homomorphism, which involves some careful analysis of the behavior
of CSM classes for joins of projective varieties.
One then defines the polynomial invariant of graphs as
$$ C_\Gamma(T)=I_{CSM}([\A^n\smallsetminus \wihat
X_\Gamma]). $$
It is natural to ask whether this polynomial invariant
may be a specialization of the Tutte polynomial.
We show in the remaining of this section that this is {\em not} the
case: the invariant $C_\Gamma(T)$ is not a specialization of the
Tutte polynomial, hence it appears to be a genuinely new invariant of graphs.

\begin{prop}\label{NoTutte}
The polynomial invariant $C_\Gamma(T)$ is not a specialization
of the Tutte polynomial.
\end{prop}

\proof We show that one cannot find functions $x=x(T)$ and
$y=y(T)$ such that 
$$ C_\Gamma(T) =  \cT_\Gamma(x(T),y(T)). $$
First notice that, if $e\in E(\Gamma)$ is a bridge, the polynomial
$C_\Gamma(T)$ satisfies the relation 
\begin{equation}\label{CGammabridge}
C_\Gamma(T)= (T+1) C_{\Gamma \smallsetminus e}(T).
\end{equation}
In fact, $(T+1)$ is the inverse propagator of the algebro-geometric
Feynman rule $U(\Gamma)=C_\Gamma(T)$ and the property of abstract
Feynman rules for 1PI graphs connected by a bridge gives
\eqref{CGammabridge}. In the case where $e\in E(\Gamma)$ is a looping
edge, we have
\begin{equation}\label{CGammaloop}
C_\Gamma(T)= T \,\, C_{\Gamma/e}(T).
\end{equation}
In fact, adding a looping edge to a graph corresponds, in terms of
graph hypersurfaces, to taking a cone on the graph hypersurface and
intersecting it with the hyperplane defined by the coordinate of the
looping edge. This implies that the universal algebro-geometric
Feynman rule with values in the Grothendieck ring $\cF$ of immersed
conical varieties satisfies $$\bU(\Gamma)=([\A^1]-1) \bU(\Gamma/e)$$
if $e$ is a looping edge of $\Gamma$ and
$\bU(\Gamma)=[\A^n\smallsetminus \wihat X_\Gamma] \in \cF$. The
property \eqref{CGammaloop} then follows since the image of
the class $[\A^1]$ is the inverse propagator $(T+1)$. (See
Proposition 2.5 and \S 2.2 of \cite{AluMa2}.)

This implies that, if $C_\Gamma(T)$ has to be a specialization of
the Tutte polynomial, the relations for bridges and looping edges
imply that one has to identify $x(T)=T+1$ and $y(T)=T$. However,
this is not compatible with the behavior of the invariant
$C_\Gamma(T)$ on more complicated graphs. For example, for the
triangle graph one has $C_\Gamma(T)=T(T+1)^2$ while the
specialization $\cT_\Gamma(x(T),y(T))=(T+1)^2+(T+1)+T$.
\endproof

The reason for this discrepancy is the fact that, while any
algebro-geometric or motivic Feynman rule will have the same behavior
as the Tutte polynomial for looping edges and bridges, the
more general deletion--contraction relation does not hold.
The class $[\A^n\smallsetminus \wihat X_\Gamma]$ in the
Grothendieck ring of varieties $K_0(\cV)$ satisfies a
more subtle deletion--contraction relation, which we now describe.

\subsection{Deletion--contraction for motivic Feynman rules}\label{delconMotSec}

We begin by considering a more general situation, which we
then specialize to the case of the graph hypersurfaces. In this
general setting, we consider two homogeneous polynomials $F$ and 
$G$ of degree $\ell-1$ and $\ell$, respectively, in variables 
$t_1,\dots,t_{n-1}$, with $n\ge 2$. Let
\begin{equation}\label{psiFG}
\psi(t_1,\ldots,t_n)=t_n F(t_1,\dots,t_{n-1})+G(t_1,\dots,t_{n-1}).
\end{equation}
Thus, $\psi$ is homogeneous of degree~$\ell$ in $t_1,\dots,t_n$.
Assume that both $F$ and~$G$ and not identically zero, so that 
it makes sense to consider the hypersurfaces defined by these polynomials. 
Cases where either $F$ or $G$ are zero are easily analyzed 
separately. We denote then by $X$ and $Y$ the projective
hypersurfaces in $\P^{n-1}$ and $\P^{n-2}$, respectively, determined
by $\psi$ and $F$. We denote by $\overline Y$ the cone of $Y$ in
$\P^{n-1}$, that is, the hypersurface defined in $\P^{n-1}$ by the
same polynomial $F$.

\begin{thm}\label{isoXYthm}
With notation as above, the projection from the point $(0:\dots:0:1)$
induces an isomorphism
\begin{equation}\label{isoXY}
X\smallsetminus (X\cap \overline Y) \overset\sim\longrightarrow 
\P^{n-2}\smallsetminus Y .
\end{equation}
\end{thm}

\proof
The projection $\P^{n-1} \dashrightarrow \P^{n-2}$ from 
$p=(0:\dots:0:1)$ acts as 
\[
(t_1:\dots:t_n) \mapsto (t_1:\dots:t_{n-1}).
\]
If $F$ is constant (that is, if $\deg \psi=1$), then 
$Y=\overline Y=\emptyset$ and the
statement is trivial. Thus, assume $\deg F>0$. In this case,
$\psi(p)=F(p)=0$, hence $p\in X\cap \overline Y$, and hence
$p\not\in X\smallsetminus (X\cap \overline Y)$. 
Therefore, the projection restricts to a regular map
\[
X\smallsetminus (X\cap \overline Y) \to \P^{n-2} .
\]
The image is clearly contained in $\P^{n-2}\smallsetminus Y$,
and the statement is that this map induces an {\em isomorphism\/}
\[
X\smallsetminus (X\cap \overline Y) \overset\sim\longrightarrow
\P^{n-2}\smallsetminus Y\quad.
\]
To see this, it suffices to verify that the (scheme-theoretic) inverse
image of any $q\in \P^{n-2}\smallsetminus Y$ is a (reduced) point
in $X\smallsetminus (X\cap \overline Y)$. Equivalently, one shows 
that the line through $p$ and $q$ meets $X\smallsetminus (X\cap \overline Y)$
transversely at one point.
Let then $q=(q_1:\dots: q_{n-1})$. The line from $p$ to $q$ is 
parametrized by 
\[
(q_1:\dots:q_{n-1}:t).
\]
Intersecting with $X$ gives the equation
\[
t F(q_1:\dots:q_{n-1})+G(q_1:\dots:q_{n-1})=0 .
\]
Since $F(q)\ne 0$, this is a polynomial of degree exactly~$1$ in $t$, and
determines a reduced point, as needed.
\endproof

This general results has some useful consequences at the level of 
classes in the Grothendieck ring $K_0(\cV)$ and of Euler
characteristic. 

\begin{cor}\label{projvers}
In the Grothendieck ring of varieties,
\begin{equation}\label{XYGro1}
[\P^{n-1}\smallsetminus X] =
[\P^{n-1}\smallsetminus (X\cap\overline Y)]
-[\P^{n-2}\smallsetminus Y] .
\end{equation}
If $\deg X>1$, then 
\begin{equation}\label{XYGro2}
[\P^{n-1}\smallsetminus X] =
\bL\cdot [\P^{n-2}\smallsetminus (Y\cap Z)]
-[\P^{n-2}\smallsetminus Y],
\end{equation}
where $\bL=[\A^1]$ is the Lefschetz motive and
$Z$ denotes the hypersurface $G=0$.
\end{cor}

\proof 
The equality \eqref{XYGro1} is an immediate consequence of Theorem~\ref{isoXYthm}.
For the second, notice that the ideal of $X\cap \overline Y$ is
\[
(\psi,F)=(t_n F+G,F)=(F,G).
\]
This means that 
\begin{equation}\label{XcapYvsYcapZ}
X\cap \overline Y=\overline Y\cap \overline Z .
\end{equation}
If $\deg X>1$, then $F$ is not constant, hence $\overline Y
\ne\emptyset$. It then follows that
$\overline Y\cap \overline Z$ contains the point $p=(0:\dots:0:1)$. 
The fibers of the projection
\[
\P^{n-1}\smallsetminus (\overline Y\cap \overline Z)
\to \P^{n-2}\smallsetminus Y
\]
with center $p$ are then all isomorphic to $\A^1$, and it 
follows that
\[
[\P^{n-1}\smallsetminus (\overline Y\cap \overline Z)]
=\bL \cdot [\P^{n-2}\smallsetminus (Y\cap Z)].
\]
This verifies the equality \eqref{XYGro2}.
\endproof

For a projective algebraic set $S\subseteq \P^{N-1}$, we
denote by $\widehat S$ the corresponding affine cone
$\widehat S\subseteq \A^N$, that is, the (conical) subset 
defined in affine space by the ideal of $S$. (Care must be
taken if $S=\emptyset$, as the corresponding cone may be the
empty set or the `origin', depending on how $S$ is defined.)

We then have the following ``affine version'' of the statement
of Corollary \ref{projvers}, where we no longer need any restriction
on $\deg X$.

\begin{cor}\label{affvers}
\begin{align*}
[\A^n\smallsetminus \widehat X] &=
[\A^n\smallsetminus \widehat {X\cap\overline Y}]
-[\A^{n-1}\smallsetminus \widehat Y] \\
&=\bL\cdot [\A^{n-1}\smallsetminus (\widehat Y\cap\widehat Z)]
-[\A^{n-1}\smallsetminus \widehat Y] .
\end{align*}
\end{cor}

\proof
If $\widehat S$ contains the origin, then it is immediately seen
that
\[
[\A^N\smallsetminus \widehat S]=(\bL-1)\cdot
[\P^{N-1}\smallsetminus S] .
\]
If $\deg X>1$, then $\deg F>0$, hence $\widehat{X\cap \overline Y}$
and $\widehat Y\cap \widehat Z$ contain the origin. In this case, 
both equalities in the statement follow from the corresponding equalities
in Corollary~\ref{projvers} by just multiplying through by $\bL-1$. 
If $\deg X=1$, then the equalities are immediately checked by hand.
\endproof

Corollary \ref{projvers} also implies 
the following relation between the Euler characteristics.

\begin{cor}\label{EulXYZ}
If $\deg X>1$, then $\chi(X)=\chi(Y\cap Z)-\chi(Y)+n$.
\end{cor}

There are interesting alternative ways to state Corollary \ref{EulXYZ}.
For example, we have the following.

\begin{cor}\label{EulXYZ2}
If $\deg X>1$, the Euler characteristics satisfy
\begin{equation}\label{EulXY1}
\chi(X\cup \overline Y)=n,
\end{equation}
or equivalently
\begin{equation}\label{EulXY2}
\chi(\P^{n-1}\smallsetminus (X\cup \overline Y))=0.
\end{equation}
\end{cor}

\proof Since $[X\smallsetminus 
(X\cap \overline Y)]=[\P^{n-2}\smallsetminus Y]$ by 
Theorem \ref{isoXYthm}, we have
\[
\chi(X)-\chi(X\cap \overline Y)=n-1-\chi(Y)
=n-\chi(\overline Y).
\]
The hypothesis $\deg X>1$ is used here, since we need $\overline
Y\ne \emptyset$. If $\deg X=1$ then one just has $\chi(X)=n-1$.
\endproof

Written in the form \eqref{EulXY2}, the statement
can also be proved by showing that there is a $\bG_m$-action
on $\P^{n-1}\smallsetminus (X\cup \overline Y)$. This is implicit
in the argument used in the proof of Theorem \ref{isoXYthm}. 

\medskip

We now consider the case of the graph hypersurfaces.

Let $\Gamma$ be a graph with $n\ge 2$ edges 
$e_1,\ldots,e_{n-1},e=e_n$, with $(t_1:\ldots :t_n)$
the corresponding variables in $\P^{n-1}$. Consider the Kirchhoff 
polynomial $\Psi_\Gamma$ and the graph hypersurface 
$X_\Gamma \subset \P^{n-1}$ as above. We can assume 
$\deg\Psi_\Gamma =\ell>0$. The case of forests can be handled separately. 
In fact, it will be occasionally convenient to assume $\deg\Psi_\Gamma
>1$, that is, assuming that the $\Gamma$ has at least two loops.

We assume that the edge $e$ is not a bridge nor a looping edge.
Here we work with arbitrary finite graphs: we do not require
that the graph is 1PI or even connected. The
Kirchhoff polynomial is still well defined. 

We then consider the polynomials
\begin{equation}\label{FGgraph}
F:=\frac{\partial \Psi_\Gamma}{\partial t_n} =\Psi_{\Gamma\smallsetminus e} 
\ \ \ \text{ and } \ \ \  
G:=\Psi_\Gamma |_{t_n=0}=\Psi_{\Gamma/e}.
\end{equation}
These are, respectively, the polynomials corresponding to the 
deletion $\Gamma\smallsetminus e$ and the contraction $\Gamma/e$ of the edge 
$e=e_n$ in~$\Gamma$. Both are not identically zero in this situation.

As above, we use the notation $\overline{Y}$ for the projective cone
over $Y$ and $\widehat Y$ for the affine cone. Then Theorem \ref{isoXYthm}
and Corollaries \ref{projvers}, \ref{affvers}, and \ref{EulXYZ}
give in this case the following deletion--contraction relations.

\begin{thm}\label{delcon}
Let $\Gamma$ be a graph with $n>1$ edges. Assume that $e$ is an edge
of~$\Gamma$ which is neither a bridge nor a looping edge.
Let $X_\Gamma$ and $\widehat X_\Gamma$ be the projective and affine 
graph hypersurfaces.
Then the hypersurface complement classes in the Grothendieck ring of
varieties $K_0(\cV)$ satisfy the deletion--contraction relation
\begin{equation}\label{delconA}
[\A^n\smallsetminus \widehat X_\Gamma]
=\bL\cdot [\A^{n-1}\smallsetminus (\widehat X_{\Gamma\smallsetminus e}
\cap \widehat X_{\Gamma/e})]
-[\A^{n-1}\smallsetminus \widehat X_{\Gamma\smallsetminus e}] .
\end{equation}
If $\Gamma$ contains at least two loops, then
\begin{equation}\label{delconP}
[\P^{n-1}\smallsetminus X_\Gamma]
=\bL\cdot [\P^{n-2}\smallsetminus (X_{\Gamma\smallsetminus e}
\cap X_{\Gamma/e})]
-[\P^{n-2}\smallsetminus X_{\Gamma\smallsetminus e}] .
\end{equation}
Under the same hypotheses, the Euler characteristics satisfy
\begin{equation}\label{delconEul}
\chi(X_\Gamma)=n+\chi(X_{\Gamma\smallsetminus e}\cap X_{\Gamma/e})
-\chi(X_{\Gamma\smallsetminus e}) .
\end{equation}
\end{thm}

The class $[\A^n\smallsetminus \widehat X_\Gamma]$ is
the universal motivic Feynman rule of \cite{AluMa2}.

In the projective case, requiring that $\Gamma$ has at least two
loops meets the condition on the degree of the hypersurface we have
in Corollary \ref{projvers}. In the one loop case, $X_\Gamma$ is
a hyperplane, so one simply gets
\[
[\P^{n-1}\smallsetminus X_\Gamma]=\bL^{n-1}\ \ \ \text{ and } \ \ \
\chi(X_\Gamma)=n-1.
\] 

The formulae for the hypersurface complement classes in
the cases where $e$ is either a bridge or a looping edge were already
covered in the results of Proposition 2.5 and \S 2.2 of
\cite{AluMa2}. We recall them here.
\begin{itemize}
\item If the edge $e$ is a bridge in $\Gamma$, then
\begin{equation}\label{bridgeGr}
[\A^n\smallsetminus \widehat X_\Gamma]
= \bL\cdot [\A^{n-1}\smallsetminus \widehat X_{\Gamma\smallsetminus e}]
= \bL\cdot [\A^{n-1}\smallsetminus \widehat X_{\Gamma/e}].
\end{equation}
In fact, if $e$ is a bridge, then $\Psi_\Gamma$ does not depend on 
the variable $t_e$ and $F\equiv 0$. The
equation for $X_{\Gamma\smallsetminus e}$ is $\Psi_\Gamma=0$ again, but viewed in
one fewer variables. The equation for $X_{\Gamma/e}$ is the
same.
\item If $e$ is a looping edge in $\Gamma$, then
\begin{equation}\label{loopGr}
[\A^n\smallsetminus \widehat X_\Gamma]
= (\bL-1)\cdot [\A^{n-1}\smallsetminus \widehat X_{\Gamma\smallsetminus e}]
= (\bL-1)\cdot [\A^{n-1}\smallsetminus \widehat X_{\Gamma/e}].
\end{equation}
In fact, if $e$ is a looping edge, then $\Psi_\Gamma$ is divisible by
$t_e$, so that $G\equiv 0$. The
equation for $X_{\Gamma/e}$ is obtained by dividing $\Psi_\Gamma$
through by~$t_e$, and one has $X_{\Gamma\smallsetminus e}=X_{\Gamma/e}$.
\end{itemize}

The formulae \eqref{delconA}, \eqref{bridgeGr}, and
\eqref{loopGr} give us the closest analog to the recursion
satisfied by the Tutte-Grothendieck invariants. Notice
that, by \eqref{XcapYvsYcapZ}, the intersection of 
$X_{\Gamma\smallsetminus e}$ and $X_{\Gamma/e}$ can in 
fact be expressed in terms of $X_{\Gamma\smallsetminus e}$ and $X_\Gamma$ alone, 
so that the result of Theorem \ref{delcon} can be expressed in 
terms that do not involve the contraction $\Gamma/e$.

\medskip

One knows from the general result of \cite{BeBro} that the 
classes $[X_\Gamma]$ of the graph hypersurfaces span the
Grothendieck ring $K_0(\cV)$ of varieties. Thus, motivically,
they can become arbitrarily complex. The question remains of
identifying more precisely, in terms of inductive procedures
related to the combinatorics of the graph, how the varieties
$X_\Gamma$ will start to acquire non-mixed Tate strata as the
complexity of the graph grows. Recent results of \cite{Doryn}
have made substantial progress towards producing explicit
cohomological computations that can identify non-mixed Tate
contributions. In the setting of deletion--contraction relations 
described above, one sees from Theorem \ref{delcon} 
that, in an inductive procedure that assembles the class of
$X_\Gamma$ from data coming from the simpler graphs $X_{\Gamma\smallsetminus e}$
and $X_{\Gamma/e}$, where one expects non-mixed Tate contributions
to first manifest themselves is in the intersection
$X_{\Gamma\smallsetminus e } \cap X_{\Gamma/e}$.

\section{Linear systems and Milnor fibers}\label{MilnorSec}

We give a different geometric interpretation of the
deletion--contraction relation proved in the previous section, which 
views the graph hypersurface
of $\Gamma$ as a Milnor fiber for hypersurfaces related to
$\Gamma\smallsetminus e$ and $\Gamma/e$. An advantage of this point of view is
that it may be better suited for extending the deletion--contraction
relation for the invariants like $C_\Gamma(T)$ defined in terms
of characteristic classes of singular varieties.

The main observation is that the deletion--contraction setting determines
a rather special linear system. With notation as above, we have
\[
\psi = t_n F(t_1,\ldots,t_{n-1}) + G(t_1,\ldots,t_{n-1}).
\]
This says that $\psi$ is in the linear system
\[
\lambda t_n F(t_1,\ldots,t_{n-1}) + \mu G(t_1,\ldots,t_{n-1}).
\]
This system specializes to $t_n F(t_1,\ldots,t_{n-1})$ for $\mu=0$ 
and to $G(t_1,\ldots,t_{n-1})$ for $\lambda=0$. 
What is special is that, for {\em every} other choice of 
$(\lambda:\mu)$, the corresponding 
hypersurface is isomorphic to $\psi=0$. Indeed, replacing
$t_n$ by $\frac \lambda\mu t_n$ gives a coordinate change
in $\P^{n-1}$ taking the hypersurface corresponding
to $(\lambda:\mu)$ to the one corresponding to $(1:1)$.

We consider the same general setting as in the previous
section, with $F$ and $G$
nonzero homogeneous polynomials of degree $\ell-1$ and $\ell$,
respectively (with $\ell>0$), in coordinates $t_1,\dots,t_{n-1}$.
We want to study the general fiber~$\psi$ of the linear system
\[
\lambda\, t_n F + \mu\, G,
\]
where we note, as above, that 
its isomorphism class is independent of the point $(\lambda:\mu)
\ne (1:0), (0:1)$.
We denote, as above, by $X\subset \P^{n-1}$ and $Y,Z\subset
\P^{n-2}$ the hypersurfaces determined by $\psi$, $F$, $G$,
respectively. We also denote by $\widehat X\subset \A^n$, 
$\widehat Y,\widehat Z \subset \A^{n-1}$ the corresponding 
affine cones.

\medskip

We can then give, using this setting, a different proof
of the statement of Corollary \ref{affvers}.

\begin{prop}\label{affversLinSys}
With the notation as above, the classes of the affine
hypersurface complements in the Grothendieck ring $K_0(\cV)$ satisfy
the deletion--contraction relation
$$ [\A^n \smallsetminus \widehat X]=\bL\cdot [\A^{n-1} \smallsetminus
(\widehat Y\cap \widehat Z)]- [\A^{n-1}\smallsetminus \widehat Y]. $$
\end{prop}
 
\proof
If $\deg X=1$, then $\widehat Y=\widehat Y\cap \widehat Z=\emptyset$.
The formula then reduces to $[\A^n\smallsetminus \A^{n-1}]=\bL\cdot 
[\A^{n-1}]-[\A^{n-1}]$, which is trivially satisfied. The formula is 
also easily checked in the case $n=2$. In fact, if $n=2$, then 
up to constants we may assume $F=t_1^{\ell-1}$ and $G=t_1^\ell$.
We can also assume $\ell>1$. We then have $\psi=t_1^{\ell-1}(t_1+t_2)$,
so that $[\widehat X]=2\bL-1$. We also have $[\widehat Y]=[\widehat Y\cap \widehat Z]
=1$. The formula then reads
\[
\bL^2-(2\bL -1)=\bL(\bL-1)-(\bL-1).
\]

We then consider the case with $\deg X>1$ and $n>2$, where we have 
$Y\ne \emptyset$ and $Y\cap Z\ne \emptyset$.
As observed above, the key to the statement is that all but two of the fibers of 
the linear system $\lambda t_n F+\mu G$ are isomorphic to~$X$. The two 
special fibers may be written as
\[
H\cup \overline Y\ \ \ \text{ and } \ \ \  \overline Z ,
\]
where $H$ is the hyperplane $t_n=0$, and $\overline Y$ and $\overline Z$
are the projective cones in~$\P^n$ over $Y$ and $Z$, respectively.
Letting $W$ denote the common intersection of all elements of the
system, we therefore have
\[
[\P^{n-1}\smallsetminus W]=(\bL-1)[X\smallsetminus W]
+[(H\cup\overline Y)\smallsetminus W]+[\overline Z\smallsetminus W],
\]
or equivalently
\[
[\P^{n-1}]=(\bL-1)\cdot [X]-\bL\cdot [W]+[H\cup \overline Y]+
[\overline Z] .
\]
Recalling that $[\A^n\smallsetminus \widehat X]=(\bL-1)
\cdot [\P^{n-1}-X]$, we get
\[
[\A^n\smallsetminus \widehat X]
=\bL\cdot [\P^{n-1}\smallsetminus W]-[\P^{n-1}\smallsetminus
(H\cup\overline Y)]-[\P^{n-1}\smallsetminus \overline Z] .
\]
Next, notice that removing the hyperplane $H$ amounts precisely
to restricting to affine space. Thus, we obtain
\[
[\P^{n-1}\smallsetminus (H\cup \overline Y)]=[\A^{n-1}
\smallsetminus \widehat Y] .
\]
As for $W=(H\cup \overline Y)\cap \overline Z\subseteq 
\P^{n-1}$, one can break up $\P^{n-1}$ as the disjoint union of 
$H=\P^{n-2}$ and~$\A^{n-1}$. Then $W$ intersects the first 
piece along $H\cap \overline Z=Z$ and the second 
along~$\widehat Y\cap \widehat Z$. Therefore, we obtain
\[
[\P^{n-1}\smallsetminus W]
=[\P^{n-2}\smallsetminus Z]+[\A^{n-1}\smallsetminus
(\widehat Y\cap \widehat Z)] .
\]
Notice that $\bL[\P^{n-2}\smallsetminus Z] 
= [\P^{n-1}\smallsetminus \overline Z]$. This shows that
\[
\bL\cdot [\P^{n-1}\smallsetminus W]
-[\P^{n-1}\smallsetminus \overline Z]
=\bL\cdot [\A^{n-1}\smallsetminus (\widehat Y\cap \widehat Z)].
\]
This completes the proof.
\endproof

In this geometric formulation one can observe also
that the projection $X_\Gamma \dashrightarrow \P^{n-2}$ 
is resolved by blowing up the point $p=(0:\dots:0:1)$,
\begin{equation}\label{blowupdiag}
\xymatrix{
& \widetilde{X_\Gamma } \ar[dl]_\nu \ar[dr]^\pi \\
X_\Gamma \ar@{-->}[rr] & &  \P^{n-2}
}
\end{equation}
The exceptional 
divisor in $\widetilde{X_\Gamma}$ is a copy of $Y=X_{\Gamma\smallsetminus e}$, mapping
isomorphically to its image in $\P^{n-2}$. The fibers
of $\pi$ are single points away from $Y\cap Z
=X_{\Gamma\smallsetminus e}\cap X_{\Gamma/e}$, and are copies of $\P^1$ over
$X_{\Gamma\smallsetminus e}\cap X_{\Gamma/e}$. In fact, $\widetilde{X_\Gamma}$ may be
identified with the blowup of $\P^{n-2}$ along the subscheme $Y \cap
Z = X_{\Gamma\smallsetminus e}\cap X_{\Gamma/e}$. This geometric setting may be
useful in trying to obtain deletion--contraction relations for
invariants defined by Chern--Schwartz--MacPherson classes, though
at present the existing results on the behavior of these classes
under blowup \cite{Alu09} 
do not seem to suffice to yield directly the desired result.

\section{Operations on graphs}\label{OpSec}

Applying the deletion--contraction formulas \eqref{delconA},
\eqref{delconP} for motivic Feynman rules obtained in Theorem
\ref{delcon} as a tool for computing the classes in the
Grothendieck ring of the graph hypersurfaces runs into a clear difficulty:
determining the intersection $X_{\Gamma\smallsetminus e}\cap X_{\Gamma/e}$. This can be
challenging, even for small graphs. In general it is bound to be, since this is
where non-Mixed-Tate phenomena must first occur. Also, this is a seemingly
`non-combinatorial' term, in the sense that it cannot be read off immediately
from the graph, unlike the ingredients in the simpler
deletion--contraction relations satisfied by the Tutte--Grothendieck 
invariants.

We analyze in this section some operations on graphs, which have
the property that the problem of describing the intersection
$X_{\Gamma\smallsetminus e}\cap X_{\Gamma/e}$ can be bypassed and the class of
more complicated graphs can be computed inductively only in terms 
of combinatorial data. The first such operation replaces a chosen
edge $e$ in a graph $\Gamma$ with $m$ parallel edges connecting the
same two vertices $\partial(e)$.

We first describe how this operation of replacing an edge in a graph
by $m$ parallel copies affects combinatorial Feynman rules such as the
Tutte polynomial. We then compare it with the behavior of the 
motivic Feynman rules under the same operation.

\subsection{Multiplying edges: the Tutte case}\label{multTutteSec}

Assume that $e$ is an edge of $\Gamma$,
and denote by $\Gamma_{me}$ the graph obtained from $\Gamma$ by
replacing $e$ by $m$ parallel edges. (Thus, $\Gamma_{0e}=\Gamma\smallsetminus e$,
and $\Gamma_e=\Gamma$.) 

Let $T_\Gamma=T(\Gamma,x,y)$ be the Tutte polynomial of the graph. 
We derive a formula for $T_{\Gamma_{me}}(x,y)$ in terms of
the polynomials for $\Gamma$ and other easily identifiable
variations. 

\begin{prop}\label{mTutteprop}
Assume $e$ is neither a bridge nor a looping edge of~$\Gamma$. Then
\begin{equation}\label{TutteGen1}
\sum_{m\ge 0} T_{\Gamma_{me}}(x,y)\, \frac {s^m}{m!}=
e^s\left(T_{\Gamma\smallsetminus e}(x,y) +\frac{e^{(y-1)s}-1}{y-1}\, 
T_{\Gamma/e}(x,y)\right)\quad.
\end{equation}
\begin{equation}\label{TutteGen2}
\sum_{m\ge 0} T_{\Gamma_{me}}(x,y)\, s^m=
\frac 1{1-s}\left(T_{\Gamma\smallsetminus e}(x,y) +\frac{s}{1-ys}\, 
T_{\Gamma/e}(x,y)\right)\quad.
\end{equation}
Explicitly, we have
\begin{equation}\label{mTutte}
T_{\Gamma_{me}}(x,y)=T_{\Gamma\smallsetminus e}(x,y)+ \frac{y^m-1}{y-1}\,
T_{\Gamma/e}(x,y) .
\end{equation}
\end{prop}

\proof If $e$ is neither a bridge nor a looping edge of $\Gamma$, then
\[
T_{\Gamma_{me}}=T_{\Gamma_{(m-1)e}}+y^{m-1} T_{\Gamma/e} .
\]
This follows from the basic recursion \eqref{Tutte2} ruling the Tutte
polynomial, observing that contracting the $m$-th copy of $e$ transforms
the first $m-1$ copies into looping edges attached to $\Gamma/e$.
Doing this recursively shows that
\[
T_{\Gamma_{me}}=T_{\Gamma\smallsetminus e}+(1+y+\dots+y^{m-1}) T_{\Gamma/e} ,
\]
which is the expression given above.

To convert this into generating functions is straightforward.
The coefficient of $T_{\Gamma\smallsetminus e}$ is immediately seen to be
as stated, in both cases. As for the coefficient of $T_{\Gamma/e}$
in the first generating function, just note that
\[
\sum_{m\ge 0} y^m\, \frac{s^m}{m!}=e^{ys}.
\]
Similarly, one has
\[
\sum_{m\ge 0} (y^m-1) s^m = \frac 1{1-ys}-\frac 1{1-s}
=(y-1)\frac s{(1-s)(1-ys)}
\]
and this gives the second generating function.
\endproof

In the case where $e$ is a bridge, one has
\[
T_\Gamma=x\,T_{\Gamma\smallsetminus e}\ \ \  \text{ and } \ \ \ 
T_{\Gamma/e}=T_{\Gamma\smallsetminus e} .
\]
Thus, everything can be written in terms of $T_{\Gamma\smallsetminus e}$.
Running through the recursion gives
\begin{align*}
T_{\Gamma_{0e}} &= T_{\Gamma\smallsetminus e} \\
T_{\Gamma_{1e}} &= T_{\Gamma} = x T_{\Gamma\smallsetminus e} \\
T_{\Gamma_{2e}} &= T_{\Gamma_{1e}}+y T_{\Gamma/e} 
= (x+y)\, T_{\Gamma\smallsetminus e} \\
T_{\Gamma_{3e}} &= T_{\Gamma_{2e}}+y^2 T_{\Gamma/e} 
= (x+y+y^2)\, T_{\Gamma\smallsetminus e} \\
&\dots
\end{align*}
This gives the generating functions 
\[
\left(e^s\left(\frac{e^{(y-1)s}-1}{y-1}+x-1\right)+2-x\right)
T_{\Gamma\smallsetminus e} ,
\]
\[
\left(\frac 1{1-s}\left(\frac s{1-ys}+x-1\right)+2-x\right)
T_{\Gamma\smallsetminus e} .
\]
The case where $e$ is a looping edge simply gives the generating
functions 
\[
e^{ys}\, T_{\Gamma\smallsetminus e}(x,y) \ \ \ \text{ and } \ \ \ 
\frac 1{1-ys}\,T_{\Gamma\smallsetminus e}(x,y).
\]

\subsection{Multiplying edges: motivic Feynman
rules}\label{multMotSec}

We now compare the behavior analyzed in the previous section in the
combinatorial setting with the case of the motivic Feynman
rules.
We use the notation as in \cite{AluMa2} for the motivic Feynman rule 
\[
\bU(\Gamma):=[\A^n\smallsetminus \Gamma] ,
\]
for $\Gamma$ a graph with $n$ edges,
with $[\A^n\smallsetminus \Gamma]$ the class of the affine
hypersurface complement in the Grothendieck ring of varieties
$K_0(\cV)$. For later use, we also introduce the notation
\begin{equation}\label{chiGamma}
\chi_\Gamma := \chi(\P^{n-1}\smallsetminus X_\Gamma),
\end{equation}
for the Euler characteristic of the projective hypersurface
complement. 

The formula in Theorem~\ref{delcon} reads then
\begin{equation}\label{motdelcon}
\bU(\Gamma)
=\bL\cdot [\A^{n-1}\smallsetminus (\widehat X_{\Gamma\smallsetminus e}
\cap \widehat X_{\Gamma/e})]
-\bU(\Gamma\smallsetminus e),
\end{equation}
under the assumption that $e$ is not a bridge or a looping
edge of $\Gamma$. We derive from this formula a multiple
edge formula in the style of those written above for the
Tutte polynomial. The nice feature these formulae exhibit 
is the fact that the complicated term $\widehat X_{\Gamma\smallsetminus e}
\cap \widehat X_{\Gamma/e}$ does
not appear and the class $\bU(\Gamma_{me})$ can be described
in terms involving only the classes $\bU(\Gamma)$,
$\bU(\Gamma\smallsetminus e)$ and $\bU(\Gamma/e)$.

By the nature of the problem, the key case is that of doubling
an edge. One obtains the following.

\begin{prop}\label{doubling}
Let $e$ be an edge of a graph $\Gamma$.
\begin{itemize}
\item If $e$ is a looping edge, then
\begin{equation}\label{double1}
\bU(\Gamma_{2e}) = \bT^2\,\bU(\Gamma\smallsetminus e) .
\end{equation}
\item If $e$ is a bridge, then
\begin{equation}\label{double2}
\bU(\Gamma_{2e}) = \bT(\bT+1)\,\bU(\Gamma\smallsetminus e) .
\end{equation}
\item If $e$ is not a bridge or a looping edge, then
\begin{equation}\label{double3}
\bU(\Gamma_{2e}) = (\bT-1)\,\bU(\Gamma) + \bT\,\bU(\Gamma\smallsetminus e)
+(\bT+1)\,\bU(\Gamma/e) ,
\end{equation}
\end{itemize}
where $\bT=[\bG_m]\in K_0(\cV)$ is the class of the
multiplicative group.
\end{prop}

\proof The formulae for the cases of a bridge or a looping edges 
follow immediately from elementary considerations, as shown
in \S 5 of \cite{AluMa1}. Thus, we concentrate on the remaining
case of \eqref{double3}, where we use the deletion--contraction
rule \eqref{motdelcon}.

Let $\Psi_{\Gamma}$, $\Psi_{\Gamma_{2e}}$ be the Kirchhoff polynomials
corresponding to the graphs $\Gamma$ and $\Gamma_{2e}$, respectively. 
We can write, as in the previous sections,
\[
\Psi_\Gamma=t_e\, F+G\quad,
\]
where $F$ is the polynomial for $\Gamma\smallsetminus e$ and $G$ is the polynomial
for $\Gamma/e$. If $e$ is replaced by the parallel edges $e$, $e'$ in
$\Gamma_{2e}$, then
\[
\Psi_{\Gamma_{2e}}=t_e t_{e'}\, F+(t_e+t_{e'})\,G
=t_{e'} (t_e F+G)+t_e G=t_{e'} \Psi_\Gamma + t_e\,G .
\]
Indeed, the term $t_e F$ in $\Psi_\Gamma$ collects the monomials 
corresponding to spanning forests that do not include $e$. The edge variable
$t_e$ is replaced by $t_e t_{e'}$ in those monomials. The term $G$ collects
monomials corresponding to spanning forests that do include $e$.
Each such monomial will appeare twice, multiplied by $t_{e'}$ when
the spanning forest is taken to include $e$, and again multiplied
by $t_e$ when the forest is taken to include $e'$.

We then apply the deletion--contraction rule to $\Psi_{\Gamma_{2e}}$,
by focusing on $e'$. Since deleting $e'$ gives us back the graph $\Gamma$,
the formula \eqref{motdelcon} gives
\begin{equation}\label{Gamma2e}
\bU(\Gamma_{2e})=\bL\cdot [\A^n\smallsetminus (\widehat X_\Gamma
\cap \widehat X_{\Gamma_o})]- \bU(\Gamma),
\end{equation}
where $n$ is the number of edges of $\Gamma$ and $\Gamma_o$ denotes
the graph obtained by attaching a looping edge named $e$ to 
$\Gamma/e$. The equation for $\Gamma_o$ is $t_e\,G$.
The ideal for this intersection is
\[
(\Psi_\Gamma,t_e\,G)\quad,
\]
so the intersection is the union of the loci defined by
\[
(\Psi_\Gamma, t_e)\ \ \ \text{ and }\ \ \  (\Psi_\Gamma, G).
\]
Simple ideal manipulations give
\[
(\Psi_\Gamma, t_e)=(t_e F+G,t_e)=(G,t_e) ,
\]
\[
(\Psi_\Gamma, G)=(t_e F+G,G)=(t_e F,G) .
\]
The latter ideal is supported on the union of the loci corresponding to
$(G,t_e)$ and $(F,G)$. The conclusion is that
\begin{equation}\label{Xo}
\widehat X_\Gamma \cap \widehat X_{\Gamma_o}
=(H\cap \widehat X'_{\Gamma/e}) \cup 
(\widehat X'_{\Gamma\smallsetminus e}\cap \widehat X'_{\Gamma/e}),
\end{equation}
where $H$ denotes the hyperplane $t_e=0$ in $\A^n$, and the primed notation
place the hypersurfaces in $\A^n$. With this notation, if
$X\subseteq \P^{n-2}$, then $\widehat X$ stands for the affine
cone over $X$, in $\A^{n-1}$, and $\widehat X'$ is the 
`cylinder' over $\widehat X$, obtained by taking the same equation 
in the larger affine space $\A^n$. We have 
$H\cap \widehat X'=\widehat X$,
and $[\widehat X']=\bL\cdot[\widehat X]$. 

By inclusion--exclusion in 
the Grothendieck ring, applied to the case of cones and 
cylinders as in \S 5 of \cite{AluMa1}, we obtain
\begin{align*}
[\widehat X_\Gamma \cap \widehat X_{\Gamma_o}]
&=[H\cap \widehat X'_{\Gamma/e}] +
[\widehat X'_{\Gamma\smallsetminus e}\cap \widehat X'_{\Gamma/e}]
-[H\cap \widehat X'_{\Gamma\smallsetminus e}\cap \widehat X'_{\Gamma/e}]\\
&=[\widehat X_{\Gamma/e}] +
(\bL-1)\cdot [\widehat X_{\Gamma\smallsetminus e}\cap \widehat X_{\Gamma/e}] .
\end{align*}
Notice that the hats on the left-hand side place the hypersurfaces
in $\A^n$, while on the right-hand side we view then in $\A^{n-1}$. This is as
it should: an affine graph hypersurface lives in a space of
dimension equal to the number of edges of the corresponding graph.

It follows then that
\begin{multline*}
[\A^n\smallsetminus (\widehat X_\Gamma \cap \widehat X_{\Gamma_o})]\\
=\bL^n 
+([\A^{n-1}- \widehat X_{\Gamma/e}]-\bL^{n-1})
+(\bL-1)\cdot (
[\A^{n-1}-(\widehat X_{\Gamma\smallsetminus e}\cap \widehat X_{\Gamma/e})]
-\bL^{n-1})\quad.
\end{multline*}
Carrying out the obvious cancellations, we get
\[
[\A^n\smallsetminus (\widehat X_\Gamma \cap \widehat X_{\Gamma_o})]
=\bU(\Gamma/e)+
(\bL-1)\cdot [\A^{n-1}-(\widehat X_{\Gamma\smallsetminus e}\cap \widehat X_{\Gamma/e})].
\]
Notice that the intersection on
the right-hand side is precisely the one that appears in the
deletion--contraction rule for $e$ on $\Gamma$. (We are using essentially here the
hypothesis that $e$ not be a bridge or a looping edge.)

Thus, we obtain 
\[
\bU(\Gamma)
=\bL\cdot [\A^{n-1}\smallsetminus (\widehat X_{\Gamma\smallsetminus e}
\cap \widehat X_{\Gamma/e})]
-\bU(\Gamma\smallsetminus e) ,
\]
So that we have
\[
\bL\cdot 
[\A^n\smallsetminus (\widehat X_\Gamma \cap \widehat X_{\Gamma_o})]
=\bL\cdot \bU(\Gamma/e) 
+(\bL-1)\cdot(\bU(\Gamma)+\bU(\Gamma\smallsetminus e)).
\]
Then plugging this into \eqref{Gamma2e} we can finally conclude
\begin{align*}
\bU(\Gamma_{2e}) &=\left(\bL\cdot\bU(\Gamma/e) 
+(\bL-1)\cdot (\bU(\Gamma)
+\bU(\Gamma\smallsetminus e))\right)-\bU(\Gamma)\\
&=(\bL-2)\cdot \bU(\Gamma)+(\bL-1)\cdot\bU(\Gamma\smallsetminus e)+
\bL\cdot\bU(\Gamma/e),
\quad,
\end{align*}
which is the statement, with $\bT=\bL-1=[\bG_m]\in K_0(\cV)$.
\endproof

A more general formula for the class of $\Gamma_{me}$ can now be
obtained using the result of Proposition~\ref{doubling}. As in 
the case of the Tutte polynomial, this is best expressed in
terms of generating functions.

\begin{thm}\label{paraledges}
Let $e$ be an edge of a graph $\Gamma$.
\begin{enumerate}
\item If $e$ is a looping edge, then
\begin{equation}\label{medge1}
\sum_{m\ge 0} \bU(\Gamma_{me})\, \frac{s^m}{m!}
=e^{\bT s}\, \bU(\Gamma\smallsetminus e) .
\end{equation}
\item If $e$ is a bridge, then
\begin{equation}\label{medge2}
\sum_{m\ge 0} \bU(\Gamma_{me})\, \frac{s^m}{m!}
=\left(\bT\cdot\frac{e^{\bT s}- e^{-s}}{\bT+1}
+ s\, e^{\bT s} +1\right)\bU(\Gamma\smallsetminus e) .
\end{equation}
\item If $e$ is not a bridge nor a looping edge, then
\begin{equation}\label{medge3} 
\begin{array}{rl}
\displaystyle{\sum_{m\ge 0} \bU(\Gamma_{me})\, \frac{s^m}{m!} }
= & \displaystyle{ \frac{e^{\bT s}-e^{-s}}{\bT+1} \bU(\Gamma) } \\[3mm]
+ & \displaystyle{ \frac{e^{\bT s}+\bT e^{-s}}{\bT+1} \bU(\Gamma\smallsetminus e) } \\[3mm]
+ & \displaystyle{ \left( s\,e^{\bT s}-\frac{e^{\bT s}-e^{-s}}{\bT+1}\right) 
\bU(\Gamma/e) }. \end{array}
\end{equation}
\end{enumerate}
\end{thm}

\proof (1) If $e$ is a looping edge, then 
$$ \bU(\Gamma_{me})= \bT^m\bU(\Gamma\smallsetminus e) , $$ 
as shown in \cite{AluMa2}, \S 2.2.

(2) For the case of a bridge, by the multiplicative
properties of abstract Feynman rules, we can write
$$ \bU(\Gamma_{me})=\epsilon_m(\bT) \bU(\Gamma\smallsetminus e), $$ 
for $m\ge 0$ and for some function $\epsilon_m$ of $\bT$, 
see Proposition~2.5 of \cite{AluMa2}. Indeed, the function
$\epsilon_m(\bT)$ is the class of the $m$-th banana graph,
which we computed explicitly in \cite{AluMa1}. 
In fact, we do not need to use the explicit
computation of $\epsilon_m(\bT)$ given in
\cite{AluMa1}, since we are going to obtain the
expression for $\epsilon_m(\bT)$ again here in a 
different way. We have
\begin{align*}
\bU(\Gamma_{0e}) &= \bU(\Gamma\smallsetminus e) \\
\bU(\Gamma_{1e}) &= (\bT+1)\cdot \bU(\Gamma\smallsetminus e) \\
\bU(\Gamma_{2e}) &= \bT(\bT+1)\cdot \bU(\Gamma\smallsetminus e)
\end{align*}
by Proposition~\ref{doubling}. For $m\ge 2$ we then have
\[
\bU(\Gamma_{(m+1)e}) = (\bT-1)\bU(\Gamma_{me})
+\bT \bU(\Gamma_{(m-1)e}) + (\bT+1)\bT^{m-1} \bU(\Gamma/e),
\]
according again to Proposition~\ref{doubling}, used to double one of
the $m$ parallel edges, which is not a bridge for $m\ge 2$.
For the third term on the right-hand side, notice that contracting
one of the $m$ parallel edges produces $m-1$ looping edges
attached to $\Gamma/e$. We then apply \cite{AluMa2}, \S2.2 to deal 
with looping edges. 
Since, in the case where $e$ is a bridge, one has $\Gamma/e=\Gamma\smallsetminus e$,
this says that
\[
\bU(\Gamma_{(m+1)e}) =
((\bT-1)\epsilon_m (\bT)+\bT\epsilon_{m-1}(\bT)
+(\bT+1)\bT^{m-1}) \bU(\Gamma\smallsetminus e)
\]
for $m\ge 2$. Thus, we obtain the family of functions 
$\epsilon_m$ as needed by solving the recurrence relation
\begin{align*}
\epsilon_0(\bT)&= 1 \\
\epsilon_1(\bT)&= \bT+1 \\
\epsilon_2(\bT)&= \bT(\bT+1) \\
\epsilon_{m+1}(\bT) &= (\bT-1)\epsilon_m(\bT)+\bT\epsilon_{m-1}(\bT)
+(\bT+1) \bT^{m-1}\ \ \ \text{ for $m\ge 2$.}
\end{align*}

Consider then the series
\begin{equation}\label{Eseries}
E(s):=\sum_{m\ge 0} \epsilon_m(\bT)\, \frac{s^m}{m!} ,
\end{equation}
so that $E(s)\cdot \bU(\Gamma\smallsetminus e)$ is the generating
function in the statement \eqref{medge2}. The recursion deals with
the coefficients $\epsilon_i$ for $i\ge 1$. It can be expressed as a relation
involving the function~$E$, taking care to truncate the first couple of terms
which are not covered by the recursion. The recursion can then be
expressed as the differential equation
\[
E''(s)-\bT(\bT+1)=(\bT-1)(E'(s)-(\bT+1))+\bT (E(s)-1)+(\bT+1)e^{\bT s}
-(\bT+1),
\]
that is,
\begin{equation}\label{eqdiffE2}
E''(s)-(\bT-1)E'(s)-\bT E(s)=(\bT+1) e^{\bT s}-\bT\quad.
\end{equation}
It is immediately checked that
\[
s\,e^{\bT s}+1
\]
is one solution of the differential equation \eqref{eqdiffE2}, and 
standard techniques show that the general solution is then of the form
\[
A\, e^{\bT s}+B\, e^{-s}+s\,e^{\bT s}+1 .
\]
Matching the initial conditions for $\epsilon_0$ and $\epsilon_1$
determines
\[
A=\frac \bT{\bT+1}\ \ \ \text{ and } \ \ \  B=-\frac \bT{\bT+1}.
\]
This yields the formula \eqref{medge2}.

(3) The situation where $e$ is not a bridge nor a looping edge is
very similar. Let
\[
\bU(\Gamma_{me}) = f_m(\bT)\, \bU(\Gamma)
+ g_m(\bT)\, \bU(\Gamma\smallsetminus e)
+ h_m(\bT)\, \bU(\Gamma/e) .
\]
These coefficients satisfy
\begin{equation}\label{fgh0}
\begin{cases}
f_0(\bT) = 0\quad, & f_1(\bT) = 1\\
g_0(\bT) = 1\quad, & g_1(\bT) = 0\\
h_0(\bT) = 0\quad, & h_1(\bT) = 0,
\end{cases}
\end{equation}
while for $m\ge 1$ the expression
\[
\bU(\Gamma_{(m+1)e})
=(\bT-1)\bU(\Gamma_{me}) +\bT \bU(\Gamma_{(m-1)e})
+(\bT+1)\bT^{m-1} \bU(\Gamma/e) 
\]
gives
\begin{align*}
\bU(\Gamma_{(m+1)e})
=&(\bT-1)(f_m(\bT)\bU(\Gamma) 
+ g_m(\bT)\, \bU(\Gamma\smallsetminus e)
+ h_m(\bT)\, \bU(\Gamma/e)) \\
&+ \bT\, (f_{m-1}(\bT)\bU(\Gamma) 
+ g_{m-1}(\bT)\, \bU(\Gamma\smallsetminus e)
+ h_{m-1}(\bT)\, \bU(\Gamma/e)) \\
&+(\bT+1) \bT^{m-1} \bU(\Gamma/e) \\
=&((\bT-1) f_m(\bT)+\bT f_{m-1}(\bT))\, \bU(\Gamma)\\
&+((\bT-1) g_m(\bT)+\bT g_{m-1}(\bT))\, \bU(\Gamma\smallsetminus e)\\
&+((\bT-1) h_m(\bT)+\bT h_{m-1}(\bT)+
(\bT+1) \bT^{m-1})\, \bU(\Gamma\smallsetminus e) .
\end{align*}
This says that the functions $f_m$, $g_m$, $h_m$ satisfy the recurrence
\[
\begin{cases}
f_{m+1} = (\bT-1) f_m + \bT f_{m-1} \\
g_{m+1} = (\bT-1) g_m + \bT g_{m-1} \\
h_{m+1} = (\bT-1) h_m + \bT h_{m-1} +(\bT+1) \bT^{m-1}
\end{cases}
\]
for $m\ge 1$. Now define the series
\[ \begin{array}{rl}
F(s):= & \displaystyle{\sum_{m\ge 0} f_m(\bT)\, \frac{s^m}{m!}},\\[3mm] 
G(s):= & \displaystyle{\sum_{m\ge 0} g_m(\bT)\, \frac{s^m}{m!}}
\\[3mm] 
H(s):= & \displaystyle{\sum_{m\ge 0} h_m(\bT)\, \frac{s^m}{m!}}, \end{array}
\]
so that
\begin{equation}\label{FGHseries}
\sum_{m\ge 0} \bU(\Gamma_{me}) \frac{s^m}{m!}
=F(s)\bU(\Gamma) + G(s) \bU(\Gamma\smallsetminus e) + H(s) \bU(\Gamma/e) .
\end{equation}
The recursions translate into the differential equations
\begin{align*}
F''(s)-(\bT-1)F'(s)-\bT F(s) &=0\\
G''(s)-(\bT-1)G'(s)-\bT G(s) &=0\\
H''(s)-(\bT-1)H'(s)-\bT H(s) &=(\bT+1) e^{\bT s} .
\end{align*}
Notice that in these cases
the recursion covers the initial indices as well, so it is not necessary to 
`truncate off' the initial terms of the series.

The homogeneous part of these equations agrees with the
homogeneous part of the equation \eqref{eqdiffE2} for $E(s)$ 
solved above. Moreover, $s\, e^{\bT s}$ is one solution of 
the third equation. Therefore, the solutions are of the form
\begin{align*}
F(s) &=A_1 e^{\bT s}+ B_1 e^{-s} \\
G(s) &=A_2 e^{\bT s}+ B_2 e^{-s} \\
H(s) &=A_3 e^{\bT s}+ B_3 e^{-s}+ s\, e^{\bT s}
\end{align*}
for suitable functions $A_i$, $B_i$ of $\bT$.
The conditions listed in \eqref{fgh0} determine these functions,
and yield the formula \eqref{medge3} given in the statement.
\endproof

\begin{rem}\label{Hirzrem} {\rm
An interesting property of the coefficients of the various
classes in the formula \eqref{medge3} of Theorem~\ref{paraledges} 
is that the quotient of the coefficients of $\bU(\Gamma)$ and 
$\bU(\Gamma\smallsetminus e)$ is the function used in defining 
Hirzebruch's $T_y$ genus, in \S 11 of Chapter III of \cite{Hirz}.
This is more evident upon rewriting the formula \eqref{medge3}
in the form
\[   \left( \frac{e^{\bT s}-e^{-s}}{\bT+1}\right) \left( \bU(\Gamma) 
+ \frac{e^{(\bT+1) s}+\bT}{e^{(\bT+1) s}-1} \bU(\Gamma\smallsetminus e) 
+ \left( \frac{(\bT+1)s}{1-e^{-(\bT+1)s}} -1 \right) 
\bU(\Gamma/e) \right) 
\]
and then comparing this expression with the formula (2) on p.94 of \cite{Hirz}.}
\end{rem}

\smallskip

We state a few direct consequences of Theorem~\ref{paraledges}. 

\begin{cor}\label{eulercharcor}
If $e$ is not a bridge or a looping edge of $\Gamma$, and 
$\Gamma\smallsetminus e$ is not a forest, then with notation as in
\eqref{chiGamma},
\[
\sum_{m\ge 0} \chi_{\Gamma_{me}} \frac {s^m}{m!}
=(1-e^{-s})\chi_\Gamma +\chi_{\Gamma\smallsetminus e}+
(s-1+e^{-s}) \chi_{\Gamma/e}
\quad.
\]
\end{cor}

\begin{proof}
This is obtained from \eqref{medge3} by dividing through by $\bT$ and
then setting $\bT=1$, since if $\Gamma$ has $n$ edges and is not a forest,
then $\bU(\Gamma) =\bT\cdot [\bP^{n-1}\smallsetminus X_\Gamma]$
(\cite{AluMa2}, Lemma~2.6).
\end{proof}

\begin{cor}\label{bananacor}
Starting with the graph $\Gamma$ that consists of a single edge (hence
a bridge), the formula \eqref{medge2} recovers the class of the
hypersurface complements of the banana graphs
\begin{equation}\label{bananaT}
 \bT\, \frac{\bT^m-(-1)^m}{\bT+1}
+ m\, \bT^{m-1}, 
\end{equation}
for $m\ge 1$, and $1$ for $m=0$.
\end{cor}

\proof Using the formula \eqref{medge2} applied to the graph consisting
of a single edge one finds that $m!$ times
the coefficient of $s^m$ in
\[
\bT\, \frac{e^{\bT s}- e^{-s}}{\bT+1}
+ s\, e^{\bT s} +1
\] 
is precisely \eqref{bananaT}.
\endproof

It is easy to obtain similar expressions for the coefficients
$\bU(\Gamma)$ and the other terms in $\bU(\Gamma_{me})$ when
 $e$ is not a bridge nor a looping edge.

\begin{cor}\label{explGme}
If $e$ is not a bridge nor a looping edge of $\Gamma$, then
\[ \begin{array}{rl}
\bU(\Gamma_{me}) = &
\displaystyle{\frac{\bT^m-(-1)^m}{\bT+1} \bU(\Gamma)}\\[3mm]
+& \displaystyle{\frac{\bT^m+(-1)^m \bT}{\bT+1} \bU(\Gamma\smallsetminus e)} \\[3mm]
+& \displaystyle{\left(m\,\bT^{m-1}-\frac{\bT^m-(-1)^m}{\bT+1}\right)
\bU(\Gamma/e)}.\end{array}
\]
\end{cor}

\proof The result follows, as in the case of Corollary
\ref{bananacor}, by reading the coefficients off the formula
\eqref{medge3} of Theorem~\ref{paraledges}.
\endproof

The first and second coefficients in are of course just alternating sums
of powers of $\bT$. One gets the second from the first
by dropping the constant term.
It is perhaps less evident that the third coefficient 
has the factorization
\[
(\bT+1)
\left( (m-1) \bT^{m-2}-(m-2) \bT^{m-3}+(m-3) \bT^{m-4}
-\cdots +(-1)^m \right).
\]
The interesting factor is the derivative of the second
coefficient. Calling
$f_m(\bT)$, $g_m(\bT)$, $h_m(\bT)$ the three coefficients,
as in the proof of Theorem~\ref{paraledges}, the statement is
that
\[
f_m=g_m-(-1)^m\quad,\quad h_m=(\bT+1)\, g'_m .
\]

Also notice that the formula \eqref{bananaT} for the banana
graph obtained in \cite{AluMa1} and in Corollary \ref{bananacor}
above, can also be described (for $m\geq 1$) in the form
\begin{equation}\label{bananaTdiff}
(\bT+1) \left(
\frac{\bT^m-(-1)^m}{\bT+1}
+\frac d{d\bT} \frac{\bT^m+(-1)^m \bT}{\bT+1}
\right).
\end{equation}

One can also formulate the result of Theorem~\ref{paraledges} in terms
of algebraic generating functions in the following form.

\begin{cor}\label{algrecur}
Let $e$ be an edge of a graph $\Gamma$.
\begin{itemize}
\item If $e$ is a looping edge, then
\[
\sum_{m\ge 0} \bU(\Gamma_{me})\, s^m
=\frac 1{1-\bT s}\, \bU(\Gamma\smallsetminus e) .
\]
\item If $e$ is a bridge, then
\[
\sum_{m\ge 0} \bU(\Gamma_{me})\, s^m
=\frac 1{(1+s)(1-\bT s)} \left(
1+s(1-\bT s)+\frac{s(1+s)}{1-\bT s}
\right)\bU(\Gamma\smallsetminus e) .
\]
\item If $e$ is not a bridge nor a looping edge, then
\[
\sum_{m\ge 0} \bU(\Gamma_{me})\, s^m
= \]\[ \frac 1{(1+s)(1-\bT s)}\left(
s\bU(\Gamma)
+(1+s-\bT s)\bU(\Gamma\smallsetminus e)
+\frac{(\bT+1) s^2}{1-\bT s}\bU(\Gamma/e)
\right).
\]
\end{itemize}
\end{cor}

\proof
These formulae are obtained by solving algebraic equations obtained from
the same recursions derived in the course of proving Theorem~\ref{paraledges},
or else directly from the explicit expressions of Corollary~\ref{explGme} and the
discussion leading to them.
\endproof

\subsection{Chains of polygons in graphs}\label{chainsSec}

As an application of the formulae obtained in Theorem~\ref{paraledges}
for parallel edges in a graph, we can provide formulae for graphs
obtained as chains of polygons. For instance, in the example given 
in Figure \ref{chainFig} one obtains that the corresponding class
$\bU(\Gamma)$ is 
$$ \bT^4(\bT+1)^{17}(\bT^3+6\bT^2+9\bT+1). $$
These graphs are inductively obtained by attaching a new polygon to
one free side of the last polygon included in the graph. It should be possible
to give similar but more involved formulae for the more general case 
in which polygons may be attached to any available free side, 
so long as no chain closes onto itself, but we only consider the
simpler class of examples here, as they suffice to illustrate the 
general principle.

\begin{figure}
\includegraphics[scale=0.5]{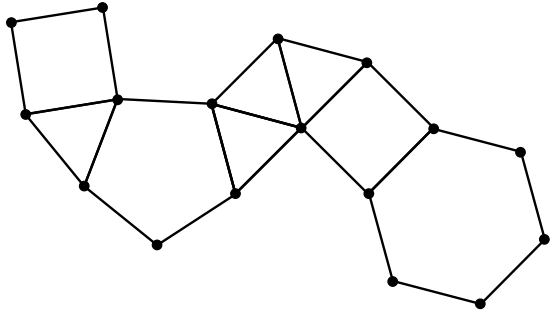}
\caption{A graph given by a chain of polygons.\label{chainFig}} 
\end{figure}

It is readily understood that, in fact, one only needs to deal with the case
in which all polygons are {\em triangles.\/} Indeed, up to isomorphism, 
the graph hypersurface is independent of the side chosen to attach the
last (and hence every) polygon: the two choices of Figure
\ref{twochoiFig} have isomorphic hypersurfaces. 

\begin{figure}
\includegraphics[scale=0.5]{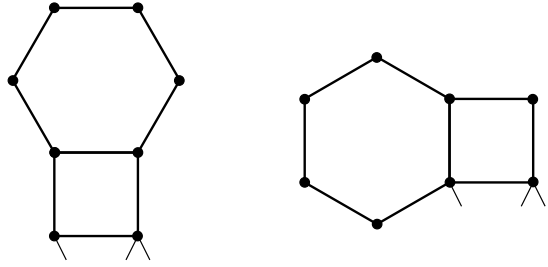}
\caption{A move on graphs which does not change the graph
hypersurface.\label{twochoiFig}} 
\end{figure}

\noindent This is because of an evident
bijection between the spanning trees of the two graphs, induced by the
switch of the two variables corresponding to the attaching edges in the
old polygon. So, for instance, the graph of Figure \ref{chain2Fig}
has graph hypersurface isomorphic to that of the one of Figure
\ref{chainFig}. 

\begin{figure}
\includegraphics[scale=0.5]{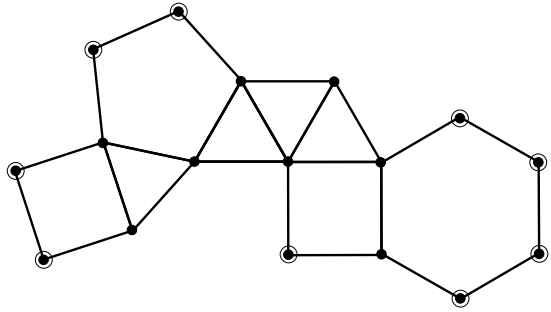}
\caption{Applying the move of Figure \ref{twochoiFig} to the graph of Figure
\ref{chainFig} does not change the graph hypersurface. \label{chain2Fig}} 
\end{figure}

Thus, we may assume that the free sides of each polygon
are all in a row. In the example of Figure \ref{chain2Fig}, 
the free vertices (marked by circles) may be obtained 
by multiple splittings of a free edge of a triangle, an operation that is 
controlled at the level of motivic invariants simply by multiplication by a 
power of $\bT+1$, since it corresponds to taking a cone (see \S 5 of
\cite{AluMa1}). Thus, all polygons in the graph of Figure
\ref{chain2Fig} may be reduced to triangles, by eliminating seven free
vertices, at the price of dividing the motivic class by a factor of
$(\bT+1)^7$. The resulting graph is illustrated in Figure
\ref{chaintriaFig}. This graph has class
\[
(\bT+1)^9\left(
\binom 80 \bT^8
+\binom 71 \bT^7
+\binom 62 \bT^6
+\binom 53 \bT^5
+\binom 44 \bT^4
\right) \]\[
=\bT^4(\bT+1)^{10}(\bT^3+6\bT^2+9\bT+1).
\]

\begin{figure}
\includegraphics[scale=0.5]{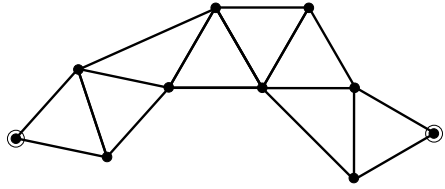}
\caption{Removing free vertices in the graph of Figure
\ref{chain2Fig}. \label{chaintriaFig}}  
\end{figure}

Since the attaching side is irrelevant, this reduces the problem of
computing the classes $\bU(\Gamma)$ of graphs obtained as chains of
polygons to that of computing the classes $\bU(\Lambda_m)$, where 
$\Lambda_m$ denotes the {\em lemon graph} with $m$ sections. For
example, the lemon graph $\Lambda_8$ of Figure \ref{lemonFig} has 
the same graph hypersurface as the graph in Figure
\ref{chaintriaFig}. 

\begin{figure}
\includegraphics[scale=0.5]{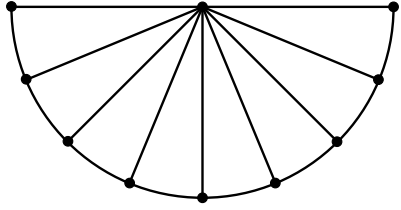}
\caption{The $8$th lemon graph $\Lambda_8$.\label{lemonFig}}
\end{figure}

The argument we described above in the example of 
Figure~\ref{chainFig} holds in general for such chains of polygons and it 
gives the following statement.

\begin{lem}\label{chainclasslem}
Let $\Gamma$ be the graph obtained as a chain of $m$ polygons with
$r_1,\dots, r_m$ sides, with $r_i\ge 3$. Then
\[
\bU(\Gamma)=(\bT+1)^{r_1+\cdots+r_m-3m}\, \bU(\Lambda_m) .
\]
\end{lem}

\proof
The indicated power simply counts the number of free vertices lost in
converting the polygons to triangles.
\endproof

A class of graphs closely related to the chains of polygons 
considered here, and their graph hypersurfaces, were recently
studied from the cohomological point of view in \cite{Doryn}.
More precisely, the type of graphs considered in \cite{Doryn},
called {\em generalized zig-zag graphs} are obtained by adding
an edge connecting the two free vertices at the ends of a
chain of triangles, in the same way in which the wheel with
$n$ spokes $W_n$ can be obtained by adding one edge connecting
the two free vertices of the lemon graph $\Lambda_n$. All 
these generalized zig-zag graphs are log divergent, like the
wheels $W_n$, which makes them especially nice from the
point of view of divergences fof the corresponding Feynman
integrals (see \cite{BroKr}, \cite{BEK}). It is proved in
\cite{Doryn} that for all these generalized zig-zag graphs,
as in the case of the wheels, the minimal non-trivial weight
piece of the Hodge structure of the corresponding projective
graph hypersurface complements is of Tate type $\Q(-2)$. 
The techniques adopted in \cite{Doryn} also involve an
analysis of the effect of removal of edges, and appear 
to be possibly related to some of our deletion--contraction
arguments.

\subsection{Lemon graphs}\label{lemonSec}

One reason why it is interesting to obtain an explicit formula
for the classes $\bU(\Lambda_m)$ of the lemon graphs, besides 
computing examples like the chain of polygons
described above, is that the $\Lambda_m$ are 
an important building block for a more complicated and more
interesting class of examples, the {\em wheel graphs} 
with $n$ spokes $W_n$ considered at length in \cite{BEK}. 

Applying the deletion--contraction relation of Theorem~\ref{delcon} 
to one spoke in the wheel $W_n$ produces the two graphs shown on the
right of Figure~\ref{wheelFig}. The class of the first 
would be known by induction, as $(\bT+1)\bU(W_{n-1})$, since the extra
vertex has the effect of taking a cone on the hypersurface hence 
multiplying the class by $(\bT+1)$, as shown in \cite{AluMa1}. 
The class of the second equals $\bU(\Lambda_n)/(\bT+1)^2$, since 
splitting the curvy edges produces the $n$-th lemon graph $\Lambda_n$.
Notice that here the class is a priori a multiple of $(\bT+1)^2$, so 
it makes sense to write $\bU(\Lambda_n)/(\bT+1)^2$.
The problem with this approach is of course that Theorem~\ref{delcon}
requires the knowledge of the class of the {\em intersection\/} of the
hypersurfaces corresponding to the two graphs on the right in 
Figure~\ref{wheelFig}, and this does not seem to be readily available.

\begin{figure}
\includegraphics[scale=0.35]{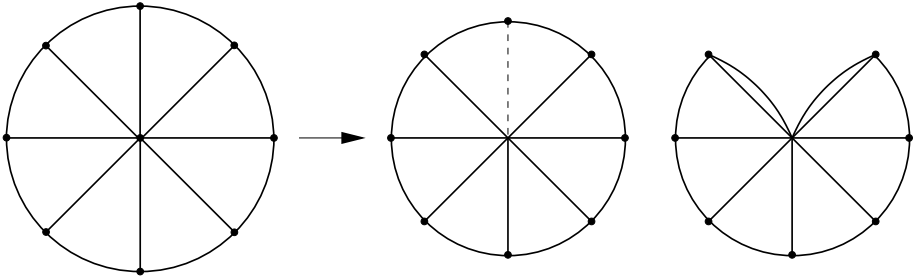}
\caption{Deletion-contraction on the wheel $W_8$.\label{wheelFig}}
\end{figure}

The classes of the lemon graphs are given by the following result,
which we formulate in terms of an algebraic generating function.

\begin{thm}\label{lemonth}
The classes $\bU(\Lambda_m)$ are determined by 
\begin{equation}\label{mLemonclass}
\sum_{m\ge 0} \bU(\Lambda_m)\, s^m=
\frac{\bT+1}{1-\bT(\bT+1)s-\bT(\bT+1)^2 s^2}.
\end{equation}
\end{thm}

\proof The theorem is proved by setting up a recursion, based on the fact that
the $(m+1)$-st lemon graph may be obtained from the $m$-th one by
doubling one edge and splitting the newly created edge, as shown in
Figure~\ref{lemonnewFig}.

\begin{figure}
\includegraphics[scale=0.3]{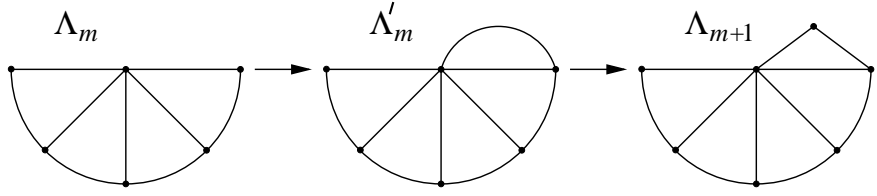}
\caption{Lemon building from edge doubling.\label{lemonnewFig}}
\end{figure}

Doubling the edge requires handling the graphs obtained by deleting
and contracting that edge as shown in Figure \ref{lemondcFig}.
These are inductively known:
\[
\bU(\Lambda_m\smallsetminus e)=(\bT+1)\bU(\Lambda_{m-1})\ \ \ \text{ and } \ \ \
\bU(\Lambda_m/e)=\bU(\Lambda_m)/(\bT+1),
\]
since adding a tail and splitting edges both have the effect of multiplying 
the motivic class by $(\bT+1)$, as shown in \cite{AluMa2}, \S 2.2.

\begin{figure}
\includegraphics[scale=0.3]{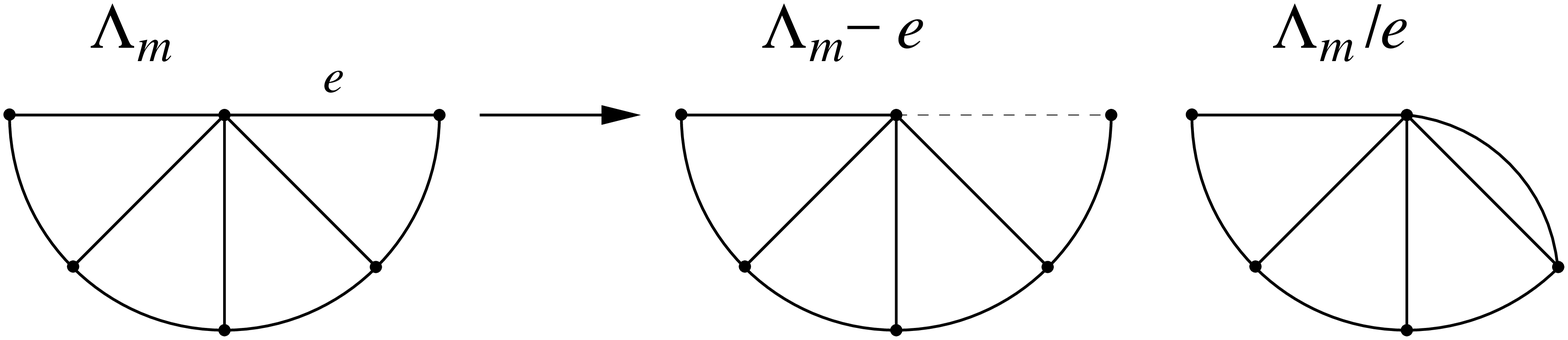}
\caption{Edge doubling in terms of deletion and contraction. \label{lemondcFig}}
\end{figure}

Applying Lemma~\ref{doubling}, with $\Lambda_m'$ denoting the second
graph of Figure~\ref{lemonnewFig}, we obtain
\begin{align*}
\bU(\Lambda_{m+1})
&=(\bT+1)\bU(\Lambda'_m) \\
& =(\bT+1)\left(
(\bT-1)\bU(\Lambda_m) + \bT \bU(\Lambda_m\smallsetminus e)+
(\bT+1)\bU(\Lambda_m/e)
\right)\\
&=(\bT+1)\left(
(\bT-1)\bU(\Lambda_m) + \bT (\bT+1) \bU(\Lambda_{m-1})+
\bU(\Lambda_m)
\right) \\
&=\bT(\bT+1) \bU(\Lambda_m) + \bT (\bT+1)^2 \bU(\Lambda_{m-1}) .
\end{align*}
This recursive relation holds as soon as the edge $e$ is not a bridge,
that is, for $m\ge 1$. The seeds are $\Lambda_0$ (a single edge)
and $\Lambda_1$ (a triangle), for which we have
\[
\bU(\Lambda_0)=\bT+1 \ \ \ \text{ and } \ \ \
\bU(\Lambda_1)=\bT (\bT+1)^2 .
\]
Let $$L_m(\bT)=\bU(\Lambda_m),$$ viewed as a polynomial in $\bT$, 
and $$L(s)=\sum_{m\ge 0} L_m s^m.$$ The recursion translates
into the relation
\[
L(s)-\bT(\bT+1)^2 s-(\bT+1) \]\[ =\bT(\bT+1) s (L(s)-(\bT+1)+\bT(\bT+1)^2
L(s) .
\]
Solving for $L(s)$ yields the formula \eqref{mLemonclass} in the statement.
\endproof

Equivalently, one can write the reciprocal of the generating 
function of \eqref{mLemonclass} of  Theorem~\ref{lemonth}, 
which has the simpler form
\[
\frac {\bT+1}{\sum_{m\ge 0} \bU(\Lambda_m) s^m}
=1-\bT (\bT+1) s-\bT (\bT+1)^2 s^2 .
\]

We then obtain from Theorem~\ref{lemonth} an 
explicit formula for the classes $\bU(\Lambda_m)$ in the
following way.

\begin{prop}\label{FiboLemon}
The classes $\bU(\Lambda_m)$ are of the form
\begin{equation}\label{FibLambdam}
 \bU(\Lambda_m) = (\bT+1)^{m+1} K(\bT), 
\end{equation}
where $K(\bT)$ is of the form
\[ \binom m0 \bT^m 
+\binom {m-1}1 \bT^{m-1}
+\binom {m-2}2 \bT^{m-2} 
+\binom {m-3}3 \bT^{m-3}
+\cdots 
\]
where $\binom ji$ is taken to be equal to $0$ if $i>j$.
\end{prop}

\proof Consider the recurrence relation
\[
a_m=a_{m-1}+x\, a_{m-2}\quad,\quad m\ge 2
\]
with $a_0=a_1=1$. This is a simple generalization of the
Fibonacci sequence, which one recovers for $x=1$.
Letting $A(t):= \sum_{m\ge 0} a_m t^m$, the recurrence gives
\[
A(t)-1-t=t(A(t)-1)+x\, t^2 A(t),
\]
hence
\[
A(t)=\frac 1{1-t-x\, t^2} .
\]
This yields an explicit expression for $a_m$: since
\[
A(t)=\sum_{k\ge 0} (1+x t)^k t^k = \sum_{k\ge 0} \sum_{i=0}^k
\binom ki t^{i+k}=\sum_{m\ge 0}\,\, \sum_{i\ge 0,i\le m-i} \binom{m-i}i
x^i t^m ,
\]
we get the expression
\[
a_m=\sum_{i= 0}^m \binom {m-i}i x^i ,
\]
adopting the convention that $\binom ji=0$ if $i>j$.
For the classes $\bU(\Lambda_m)$ of the lemon graphs we have from
Theorem \ref{lemonth} the generating function
\[
\frac{\bT+1}{1-\bT(\bT+1) s-\bT(\bT+1)^2 s^2}
=(\bT+1) \frac 1{1-(\bT(\bT+1) s)-\frac 1\bT (\bT(\bT+1) s)^2} .
\]
Thus, upon setting $t=\bT(\bT+1)s$ and $x= 1/\bT$, the previous
considerations give
\[
\bU(\Lambda_m) =
(\bT+1) \left(\sum_{i\ge 0} \binom {m-i}i \frac 1{\bT^i} \right)
\bT^m (\bT+1)^m \] \[
= (\bT+1)^{m+1} \sum_{i= 0}^m \binom {m-i}i \bT^{m-i},
\]
which gives \eqref{FibLambdam}.
\endproof

As we have seen in the proof of Proposition \ref{FiboLemon} above,
the classes $\bU(\Lambda_m)$ are closely related to a Fibonacci-like
recursion. In fact, they satsify the following property, which is
the analog of the well known property of Fibonacci numbers.

\begin{cor}\label{LemonDivis}
The sequence $a_m=\bU(\Lambda_{m-1})$ is a divisibility sequence.
\end{cor}

\proof A sequence $a_m$ is a divisibility sequence if $a_m | a_n$
whenever $m|n$. We show that the expression for $\bU(\Lambda_{m-1})$ 
divides the expression for $\bU(\Lambda_{n-1})$ if $m$ divides~$n$.
Using the recursion relation, this follows by showing that if 
\[
\frac t{1-t-xt^2} = \sum_{n\ge 0} b_n(x) t^n
\]
then, if $m$ divides $n$, then the polynomial $b_m(x)$ divides 
the polynomial $b_n(x)$. The $t$ in the numerator produces the
shift of one in the indices. 

The polynomials
$$ b_n(x) = a_{n-1}(x)=\sum_{i= 0}^{n-1} \binom{n-1-i}i x^i $$
can also be written in the form
\begin{equation}\label{bmlambda}
 b_m(x)=\frac{\lambda_1^m
-\lambda_2^m}{\lambda_1-\lambda_2} 
\end{equation}
where 
\[
\lambda_1=\frac{1+\sqrt{1+4x}}{2}\ \ \ \text{ and } \ \ \ 
\lambda_2=\frac{1-\sqrt{1+4x}}{2} .
\]
Then
\[
\frac{b_{km}(x)}{b_m(x)}=\frac{\lambda_1^{km}-\lambda_2^{km}}
{\lambda_1^m-\lambda_2^m}
=\lambda_1^{(k-1)m} +\cdots +\lambda_2^{(k-1)m}  \]
is clearly a polynomial. One can explicitly provide a 
recurrence relation satisfied by the function of~$k$ given by $b^{(m)}_k
=b_{km}(x)/b_m(x)$. First note that $\lambda_1^m$, $\lambda_2^m$
are roots of a quadratic polynomial
\[
(y-\lambda_1^m)(y-\lambda_2^m)=y^2-T Y +N=0 ,
\]
where $T=\lambda_1^m+\lambda_2^m$ and $N=(\lambda_1\lambda_2)^m=(-x)^m$. 
Notice then that 
\[
\sum_{m\ge 0} (\lambda_1^m+\lambda_2^m) t^m = 
\frac 1{1-\lambda_1 t}+\frac 1{1-\lambda_2 t} \]\[ = 
\frac{2-t}{1-t+t^2}=\sum_{m\ge 0}(2b_{m+1}(x)-b_m(x)) t^m .
\]
This shows that $T=2b_{m+1}(x)-b_m(x)$. Therefore
\[
(y-\lambda_1^m)(y-\lambda_2^m)=y^2-(2b_{m+1}(x)-b_m(x)) y +(-x)^m\quad.
\]
It follows that $b^{(m)}_k(x):=\frac{b_{km}(x)}{b_m(x)}
=\frac{\lambda_1^{km}-\lambda_2^{km}}{\lambda_1^m-\lambda_2^m}$
are solutions of the recurrence relation
\[
b^{(m)}_{k+1}=(2b_{m+1}(x)-b_m(x)) b^{(m)}_k - (-x)^m b^{(m)}_{k-1}\quad,
\]
with seeds $b^{(m)}_0=0$, $b^{(m)}_1=1$.
\endproof

In terms of understanding explicitly the motivic nature of the graph
hypersurfaces for certain infinite families of graphs, the result of
Theorem~\ref{lemonth}, together with Lemma~\ref{chainclasslem}, has
the following direct consequence.

\begin{cor}\label{MTMchains}
All graphs $\Gamma$ that are polygon chains have graph hypersurfaces
$X_\Gamma$ whose classes $[X_\Gamma]$ in the Grothendieck ring are
contained in the Tate subring $\Z[\bL]\subset K_0(\cV)$.
\end{cor}

\subsection{Graph lemonade}\label{lemonadeSec}

As a variation on the same theme explored here, one can compute
the class of the graph obtained from {\em any\/} graph $\Gamma$
by `building a lemon' on a given edge $e$, as in Figure
\ref{lemonadeFig}. 

\begin{figure}
\includegraphics[scale=0.45]{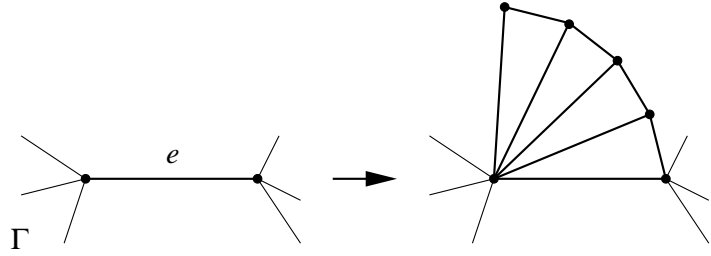}
\caption{Adding a lemon to a graph.\label{lemonadeFig}}
\end{figure}

The question makes no sense if $e$ is a looping edge, and is 
covered by multiplicativity if $e$ is a bridge, so we can 
assume that $e$ is not either. One obtains then a formula 
expressing the class of the ``lemonade'' of the graph $\Gamma$
at the edge $e$ in terms of $\bU(\Gamma)$, $\bU(\Gamma\smallsetminus e)$, 
$\bU(\Gamma/e)$.

\begin{prop}\label{lemonth2}
Let $e$ be an edge of a graph $\Gamma$, and assume that $e$ is
neither a bridge nor a looping edge. 
Let $\Gamma^\Lambda_m$ be the ``lemonade graph'' obtained 
by building an $m$-lemon fanning out from $e$. Then
\begin{multline*}
\sum_{m\ge 0}\bU(\Gamma^\Lambda_m) s^m
=\frac 1{1-\bT(\bT+1)s-\bT(\bT+1)^2 s^2} \\
\cdot \left((1-(\bT +1)s)\, \bU(\Gamma)
+(\bT+1)\bT s\, \bU(\Gamma\smallsetminus e)
+(\bT+1)^2s\, \bU(\Gamma/e) \right) .
\end{multline*}
\end{prop}

\proof
Let $f_m$, $g_m$, $h_m$ be functions of $\bT$ such that
\[
\bU(\Gamma^\Lambda_m)
=f_m(\bT) \bU(\Gamma)
+g_m(\bT) \bU(\Gamma\smallsetminus e)
+h_m(\bT) \bU(\Gamma/e) .
\]
The basic recursion is precisely the one worked out in the proof of
Theorem~\ref{lemonth}. Namely,
\[
\bU(\Lambda_{m+1})
=\bT(\bT+1) \bU(\Lambda_m) + \bT (\bT+1)^2 \bU(\Lambda_{m-1}).
\]
This makes sense for $m\ge 1$. The individual functions $f_m$, $g_m$,
$h_m$ satisfy this same recursion, but with different seeds:
\[
\begin{cases}
f_0(\bT)=1\quad, & f_1(\bT)=\bT^2-1 \\
g_0(\bT)=0\quad, & g_1(\bT)=\bT(\bT+1) \\
h_0(\bT)=0\quad, & h_1(\bT)=(\bT+1)^2 .
\end{cases} 
\]
The values in the second column implement the doubling formula of
Lemma~\ref{doubling}, and then split the new edge by introducing a
factor of $\bT+1$. Letting $F(s)$, $G(s)$, $H(s)$ be the three
corresponding generating functions $F(s)=\sum_{m\ge 0} f_m s^m$,
etc., the recursions imply
\begin{align*}
F(s)-(\bT^2-1)s-1 &= \bT(\bT+1) s(F(s)-1) + \bT(\bT+1)^2 s^2 F(s) \\
G(s)-\bT(\bT+1)s &= \bT(\bT+1) s G(s) + \bT(\bT+1)^2 s^2 G(s) \\
H(s)-(\bT+1)^2s &= \bT(\bT+1) s H(s) + \bT(\bT+1)^2 s^2 H(s)
\end{align*}
from which
\begin{align*}
F(s) &= \frac{1-(\bT +1)s}{1-\bT(\bT+1)s-\bT(\bT+1)^2 s^2} , \\
G(s) &= \frac{(\bT+1)\bT s}{1-\bT(\bT+1)s-\bT(\bT+1)^2 s^2} , \\
H(s) &= \frac{(\bT+1)^2 s}{1-\bT(\bT+1)s-\bT(\bT+1)^2 s^2} ,
\end{align*}
as stated.
\endproof

The three functions $F(s)$, $G(s)$, $H(s)$ are all easily recoverable
from the lemon formula of Theorem~\ref{lemonth}.

\smallskip

Notice, moreover, that the result of Proposition~\ref{lemonth2} yields
immediately the following generating series for the Euler characteristic
of the complement of~$X_{\Gamma^\Lambda_m}$, valid under the same hypotheses
of the proposition. If $\Gamma\smallsetminus e$
is not a forest, then
\[
\sum_{m\ge 0} \chi_{\Gamma^\Lambda_m} s^m
=(1- s)\, \chi_\Gamma +\chi_{\Gamma/e} .
\]
That is, $\chi_{\Gamma^\Lambda_1}=\chi_{\Gamma/e}-\chi_{\Gamma}$
and $\chi_{\Gamma^\Lambda_m}=0$ for $m>1$.

\section{Universal recursion relation}\label{UnivSec}

The very structure of the problems analyzed in the previous section is
recursive, and this fact alone is responsible for some of the features
of the solutions found in \S\ref{OpSec}.
We emphasize these general features in this section, and apply them
to formulate a precise conjecture for the effect of the operation of
multiplying edges on the polynomial invariant $C_\Gamma(T)$
of graphs obtained in \cite{AluMa2} in terms of CSM classes; recall that
we have shown in \S \ref{CSMsec} that this is not a specialization of the 
Tutte polynomial. 

\subsection{Recursions from multiplying edges}\label{genrecu}
Let $\Gamma$ be a graph with two (possibly coincident) marked vertices $v$, $w$
in the same connected component. Typically, the vertices will be the
boundary of an edge $e$ of $\Gamma$. We consider the 
operation $\Gamma \leadsto \Gamma^{(m)}$ which has the effect
of inserting $m$ parallel edges joining $v$ and $w$.
Note that obviously $\Gamma^{(m+n)}=(\Gamma^{(m)})^{(n)}$.
A feature of invariants
$U$ such as the Tutte polynomial and the motivic Feynman rule $\bU$ is
that if $e$ is an edge joining $v$ and $w$ in~$\Gamma$, so that
$\Gamma^{(m)}=\Gamma_{(m+1)e}$ with the notation of \S
\ref{multTutteSec} and \S \ref{multMotSec}, then the effect 
of this operation on $U$ can be expressed
{\em consistently\/} as
\[
U(\Gamma^{(m)})=f_{m+1} U(\Gamma) + g_{m+1} U(\Gamma\smallsetminus e)
+ h_{m+1} U(\Gamma/e)\quad.
\]
The consistency requirement may be formulated as follows. Let $\cR$
be the target of the invariant $U$, and consider the evaluation map
$\cR^{\oplus 3} \to \cR$ given by
\[
\begin{pmatrix}
g \\ f \\ h 
\end{pmatrix}
\mapsto g\, U(\Gamma\smallsetminus e) + f\, U(\Gamma) + h \, U(\Gamma/e)
\quad;
\]
then the main requirement is that the effect of $\Gamma \leadsto \Gamma^{(m)}$
can be lifted to a representation of the additive monoid $\bZ^{\ge 0}$ on
$\cR^{\oplus 3}$. In other words, these operations may be 
represented by $3\times 3$ matrices $A_m$, such that
\[
A_m \cdot
\begin{pmatrix} 
0 \\ 1 \\ 0
\end{pmatrix}
= \begin{pmatrix} 
g_{m+1} \\ f_{m+1} \\ h_{m+1}
\end{pmatrix}
\mapsto
 f_{m+1} U(\Gamma) + g_{m+1} U(\Gamma\smallsetminus e)
+ h_{m+1} U(\Gamma/e)
=U(\Gamma^{(m)})
\]
and $A_0=1$, $A_{m+n}=A_m\cdot A_n$. It is also natural to assume
that 
\[
A_1 \cdot
\begin{pmatrix} 
0 \\ 1 \\ 0
\end{pmatrix} = 
\begin{pmatrix} 
1 \\ 0 \\ 0
\end{pmatrix}
\quad,
\]
reflecting the fact that $\Gamma=(\Gamma\smallsetminus e)^{(1)}$ (assuming
$\Gamma\smallsetminus e$ is marked by $\partial e$), and 
\[
A_1 \cdot
\begin{pmatrix} 
0 \\ 0 \\ 1
\end{pmatrix} = 
\begin{pmatrix} 
0 \\ 0 \\ Z
\end{pmatrix}
\quad,
\]
where $Z$ is the value of $U$ on the graph consisting of a single
looping edge. This last requirement is  
motivated by the fact that the endpoints $v$, $w$ of~$e$ coincide
in the contraction $\Gamma/e$, therefore $(\Gamma/e)^{(1)}$
consists of $\Gamma/e$ with a looping edge attached to $v=w$;
since $U$ is a Feynman rule, the effect must amount to simple
multiplication by $Z$.

The datum of the representation is captured by the generating
function
\[
A(s):=\sum_{m\ge 0} A_m \frac{s^m}{m!} = e^{A_1 s} \quad.
\]
Our task is to determine this function, or equivalently the generating
functions
\[
f(s)=\sum_{m\ge 0} f_m \frac{s^m}{m!} \quad, \quad
g(s)=\sum_{m\ge 0} g_m \frac{s^m}{m!} \quad, \quad
h(s)=\sum_{m\ge 0} h_m \frac{s^m}{m!} \quad.
\]
Here, it is natural to set
\begin{equation}
\label{initconds}
\begin{cases}
f_0 = 0 \\
f_1 = 1 
\end{cases}
\quad,\quad
\begin{cases}
g_0 = 1 \\
g_1 = 0 
\end{cases}
\quad,\quad
\begin{cases}
h_0 = 0 \\
h_1 = 0 
\end{cases}
\quad.
\end{equation}

\begin{lem}
Let $U$ be a Feynman rule, and assume that $Z$ is the value of $U$ on 
the graph consisting of a single looping edge. Then for $m\ge 0$
\[
A_m=\begin{pmatrix}
g_m & g_{m+1} & 0 \\
f_m & f_{m+1} & 0 \\
h_m & h_{m+1} & Z^m
\end{pmatrix}
\]
and the coefficients $f_m$, $g_m$, $h_m$ satisfy the following recursion
\begin{equation}
\label{recursionfghA}
\begin{cases}
f_{m+2}=f_2 f_{m+1}+g_2 f_m \\
g_{m+1}=g_2 f_m \\
h_{m+1}=h_2 f_m+Z h_m
\end{cases}
\end{equation}
for $m\ge 0$.
\end{lem}

\begin{proof}
By assumption, 
\[
A_m\cdot \begin{pmatrix}
0 \\ 1 \\ 0
\end{pmatrix}
=\begin{pmatrix}
g_{m+1} \\ f_{m+1} \\ h_{m+1}
\end{pmatrix}
\quad\text{and}\quad
A_1=\begin{pmatrix}
0 & g_2 & 0\\
1 & f_2 & 0\\
0 & h_2 & Z
\end{pmatrix}
\quad.
\]
Assuming inductively that
$
A_{m-1}=\begin{pmatrix}
g_{m-1} & g_m & 0 \\
f_{m-1} & f_m & 0 \\
h_{m-1} & h_m & Z^m
\end{pmatrix}
$,
the fact that $A_m=A_{m-1} A_1$ shows that $A_m$ has the stated
shape. The recursion is
forced by the fact that $A_{m+1}=A_1A_m$, which gives
\begin{align*}
\begin{pmatrix}
g_{m+1} & g_{m+2} & 0\\
f_{m+1} & f_{m+2} & 0 \\
h_{m+1} & h_{m+2} & Z^{m+1}
\end{pmatrix}
&=\begin{pmatrix}
0 & g_2 & 0\\
1 & f_2 & 0\\
0 & h_2 & Z
\end{pmatrix}
\begin{pmatrix}
g_m & g_{m+1} & 0\\
f_m & f_{m+1} & 0\\
h_m & h_{m+1} & Z^m
\end{pmatrix}\\
&=
\begin{pmatrix}
g_2 f_m & g_2 f_{m+1} & 0\\
g_m+f_2 f_m & g_{m+1}+ f_2 f_{m+1} & 0\\
h_2 f_m + Z h_m & h_2 f_{m+1}+ Z h_{m+1} & Z^{m+1}
\end{pmatrix}
\end{align*}
and shows that
\[
\begin{cases}
f_{m+2} =f_2 f_{m+1} +g_{m+1}\\
g_{m+1} =g_2 f_m \\
h_{m+1} = h_2 f_m + Z h_m 
\end{cases} .
\]
Then (\ref{recursionfghA}) follows immediately.
\end{proof}

In specific cases, the recursion can often be solved by computing explicitly
$e^{A_1 s}$, which is straightforward if $A_1$ is diagonalizable. This
can be carried out easily for the the motivic Feynman rule, for which
\[
A_1=
\begin{pmatrix}
0 & \bT-1 & 0 \\
1 & \bT & 0 \\
0 & \bT+1 & \bT
\end{pmatrix}
\quad,
\]
recovering formula \eqref{medge3} in Theorem~\ref{paraledges}. It can 
also be worked out for the Tutte polynomial, for which we can choose
\begin{equation}
\label{Tutch}
A_1=
\begin{pmatrix}
0 & 1 & 0 \\
1 & 0 & 0 \\
0 & 1+y & y
\end{pmatrix} .
\end{equation}
Since the Tutte polynomial satisfies the relation $T_\Gamma=
T_{\Gamma\smallsetminus e}+T_{\Gamma/e}$, there are in fact many possible choices
for the corresponding representation. The one chosen in \eqref{Tutch} gives
\[
A(s)=\sum_{m\ge 0} A_m \frac {s^m}{m!}
=\begin{pmatrix}
0 & \frac{1-y}{1+y} & 1 \\
0 & \frac{1-y}{1+y} & -1 \\
1 & 1 & 1
\end{pmatrix}
\begin{pmatrix}
e^{ys} & 0 & 0 \\
0 & e^s & 0\\
0 & 0 & e^{-s}
\end{pmatrix}
\begin{pmatrix}
\frac 1{y-1} & \frac y{y-1} & 1 \\
\frac 12\frac{1+y}{1-y} & \frac 12\frac{1+y}{1-y} & 0 \\
\frac 12 & -\frac 12 & 0
\end{pmatrix}
\]
and correspondingly
\[ \begin{array}{rl}
f(s)= & \displaystyle{\frac{e^s-e^{-s}}2}=\sinh s \\[3mm]
g(s)= & \displaystyle{\frac{e^s+e^{-s}}2}=\cosh s \\[3mm]
h(s)= & \displaystyle{\frac{e^{ys}-e^s}{y-1}-\sinh s} . \end{array}
\]
Since the deletion--contraction relation \eqref{Tutte2} holds, 
this is equivalent to the result of Proposition~\ref{mTutteprop}.

The recursion (\ref{recursionfghA}) can be solved directly in general
by the same method for the specific cases analyzed in \S\ref{OpSec}.
The conditions translate into differential equations satisfied by the
functions $f$, $g$, $h$, and specifically
\begin{equation}
\label{diffeq}
\left\{
\begin{aligned}
f''(s)&=f_2\, f'(s)+g_2 f(s) \\
g'(s)&=g_s\, f(s) \\
h'(s)&=Z\, h(s)+h_2\, f(s)
\end{aligned}
\right.
\quad.
\end{equation}
With the initial conditions specified in (\ref{initconds}), and assuming
$f_2^2+4g_2\ne 0$, the first equation
has the solution
\[
f(s)=\frac{e^{\lambda_+ s}-e^{\lambda_- s}}{\lambda_+-\lambda_-}\quad,
\quad
\text{where $\lambda_\pm=\frac{f_2\pm \sqrt{f_2^2+4 g_2}}2$}\quad;
\]
equivalently,
\[
f_m=\frac{\lambda_+^m-\lambda_-^m}{\lambda_+-\lambda_-}
=\lambda_+^{m-1}+\lambda_+^{m-2}\lambda_-+\cdots +\lambda_-^{m-1}
\quad.
\]
The second and third equations then determine $g$ and $h$:
\[
g(s)=\frac{\lambda_+ e^{\lambda_- s}- \lambda_- e^{\lambda_+ s}}
{\lambda_+-\lambda_-}
\]
and
\[
h(s)=\frac {h_2}{\lambda_+-\lambda_-}
\left(\frac{e^{\lambda_+s}-e^{Zs}}{\lambda_+-Z}-
\frac{e^{\lambda_-s}-e^{Zs}}{\lambda_--Z}\right)
\]
if $\lambda_+\ne Z$ and $\lambda_-\ne Z$, while
\[
h(s)=\frac {h_2}{Z-\lambda}
\left(s\,e^{Zs}-\frac{e^{Z s}-e^{\lambda s}}{Z-\lambda}\right)
\]
if $\{\lambda_+,\lambda_-\}=\{Z,\lambda\}$. This last eventuality occurs
for the motivic Feynman rule.

The equations (\ref{diffeq}) highlight interesting features of the
coefficients of any solution to the multiplying edge problem, independent
of the specific context. From the general solution, we also see that
\[
\frac{f(s)}{g(s)}=\frac{e^{\lambda_+ s}-e^{\lambda_- s}}
{\lambda_+ e^{\lambda_- s}-\lambda_- e^{\lambda_+ s}}
\quad,
\]
generalizing Remark~\ref{Hirzrem}. For the Tutte polynomial (with the choice
of~\eqref{Tutch}) this function is the hyperbolic tangent.

As a last general remark, we note that the coefficients $f_m$ form
a divisibility sequence. This is clear from the expression for $f_m$
given above: $f^{(m)}_r:=\frac{f_{rm}(s)}{f_m(s)}=\frac{\lambda_+^{rm}
-\lambda_-^{rm}}{\lambda_+^m-\lambda_-^m}$.
Alternatively, it can be proved as in Corollary~\ref{LemonDivis}:
one finds that the quotients $f^{(m)}_r$ satisfy the recursion
\[
f^{(m)}_{r+2}=(f_2 f_m + 2 g_2 f_{m-1}) f^{(m)}_{r+1} - (-g_2)^m f^{(m)}_r
\quad,
\]
for all $r\ge 0$, and in particular it follows that $f_{m}(s)$ divides $f_{rm}(s)$
for all~$r\ge 0$.

\subsection{Conjectural behavior of $C_\Gamma$ under multiplication 
of edges}
A deletion/contraction rule for the invariant $C_\Gamma$ is not yet
available; however, if such a rule exists then a doubling-edge formula
for this invariant should exist, of the type considered above: if $e$ 
is an edge of $\Gamma$ joining the marked vertices, then one would
expect
\[
C_{\Gamma^{(m)}}\overset ?=
f_{m+1}\cdot C_{\Gamma}+ g_{m+1}\cdot C_{\Gamma\smallsetminus e}+h_{m+1}
\cdot C_{\Gamma/e}
\]
for suitable coefficients satisfying the stringent requirements examined
in~\S\ref{genrecu}. The fact that $C_\Gamma$ is known for banana
graphs (Example~3.8 in \cite{AluMa2}) provides then a testing ground for
this phenomenon, as well as a precise indication for what the 
needed representation should be in this case.

\begin{guess}\label{CTme}
The polynomial Feynman rule $C_\Gamma$ obeys the general
recursion formulas obtained in \S\ref{genrecu} with respect to the operation
of multiplying edges. The corresponding representation is determined by
\[
f_2 = 2T -1, \ \ \ \ g_2 = -T(T-1), \ \ \ \ h_2 =1\quad.
\]
The generating functions for the operation are
$$ \begin{array}{rl}
f(s) = & e^{Ts}- e^{(T-1)s} \\[3mm]
g(s) = & T e^{(T-1)s} - (T-1) e^{Ts} \\[3mm]
h(s) = & e^{(T-1)s} + (s-1) e^{Ts}.
\end{array} $$
\end{guess}

Since the Euler characteristic $\chi_\Gamma$ of the complement of 
$X_\Gamma$ in its ambient projective space can be recovered from
$C_\Gamma$ (cf.~Proposition~3.1 in \cite{AluMa2}), these formulae
imply generating functions for $\chi_{\Gamma^{(m)}}$. These coincide
with the formulae obtained in Corollary~\ref{eulercharcor}, providing some
evidence for Conjecture~\ref{CTme}.

Conjecture~\ref{CTme} is verified for all cases known to us: the family
of banana graphs, as well as several examples for small graphs computed
by J.~Stryker~(\cite{Stryker}). In fact, the smallest graph for which the
invariant $C_\Gamma(T)$ is not known is the triangle with doubled edges;
according to Conjecture~\ref{CTme}, the polynomial invariant for this
graph is
\[
T^6+2T^5+8T^4+2T^3+T^2-T\quad,
\]
and it follows that the CSM class of the corresponding
graph hypersurface should be
\[
4[\bP^4]+7[\bP^3]+18[\bP^2]+14[\bP]+7[\bP^0]\quad.
\]

\section{Categorification}\label{CategSec}

Various examples of categorifications of graph and link invariants
have been recently developed. These are categorical constructions with associated
(co)homology theories, from which the (polynomial) invariant is
recontructed as Euler characteristic. The most famous examples of
categorification are Khovanov homology \cite{Khovanov}, which is a 
categorification of the Jones polynomial, and graph homology, which
gives a categorification of the chromatic polynomial \cite{HeGuRo}. 
More recently, a categorification of the Tutte polynomial was also
introduced in \cite{JaHeRo}.
In the known categorifications of invariants obtained from
specializations of the Tutte polynomial, the deletion--contraction
relations manifest themselves in the form of long exact (co)homology
sequences.
Another way in which the notion of categorification found applications
to algebraic structures associated to graphs is in the context of
Hall algebras. In this context, one looks for a categorification of
a Hopf algebra, that is, an abelian category such that the given
Hopf algebra is the assoctaed Hall--Ringel algebra. In the case
of the Connes--Kreimer Hopf algebra of Feynman graphs, a suitable
categorification, which realizes it (or rather its dual Hopf algebra)
as a Ringel--Hall algebra, was recently obtained in \cite{KreSz}.
 
In view of all these results on categorification, it seems natural to
try to interpret the deletion--contraction relation described in
this paper for the motivic Feynman rules in terms of a suitable
categorification. As remarked in \S 8 of \cite{Blo}, one can think
of the motive associated to the graph $\Gamma$ and the maps induced
by edge contractions as a motivic version of graph cohomology. 
We see a similar setting here in terms of the deletion--contraction
relations we obtained in \S \ref{HypersurfSec} and \S
\ref{MilnorSec}. 

We denote by $\fm(X)$ the motive of a variety $X$, seen 
as an object in the triangulated category $\cD\cM_\Q$ 
of mixed motives of \cite{Voe}. A closed 
embedding $Y\subset X$ determines a distinguished
triangle in this category
\begin{equation}\label{distTri}
\fm(Y) \to \fm(X) \to \fm(X\smallsetminus Y) \to \fm(Y)[1].
\end{equation}
Since one thinks of motives as a universal cohomology theory for
algebraic varieties, and of classes in the Grothendieck ring as a
universal Euler characteristic, it is natural to view the motive
$\fm(X_\Gamma)\in \cD\cM_\Q$ as the ``categorification'' of the ``Euler
characteristic'' $[X_\Gamma]\in K_0(\cV_\Q)$.

The analog of the fact that the categorification of
deletion--contraction relations takes the form of long exact
cohomology sequences is then expressed in this context in the
following way.

\begin{prop}\label{distTdelcon}
For a graph $\Gamma$ with $n$ edges, 
let $\fm_\Gamma := \fm(\P^{n-1}\smallsetminus X_\Gamma)$ as an object
in $\cD\cM_\Q$.
The deletion contraction relation of Theorem~\ref{delcon} 
determines a distinguished triangle in $\cD\cM_\Q$ of the form
\begin{equation}\label{delconDM}
\fm_{\Gamma\smallsetminus e} \to 
\fm_\Gamma \to \fm(\P^{n-1}\smallsetminus 
(\overline{X_{\Gamma\smallsetminus e}} \cap \overline{X_{\Gamma/e}})) \to \fm_{\Gamma\smallsetminus e}[1], 
\end{equation}
where, as above $\overline{X}$ denotes the cone on $X$.
\end{prop}

\proof This follows from the proof of Theorem \ref{delcon}. In fact,
for the inclusion $X_\Gamma \smallsetminus (X_\Gamma \cap
\overline{X_{\Gamma\smallsetminus e}})$ in $\P^{n-1}\smallsetminus X_\Gamma$, we
have a distringuished triangle of the form
$$ \fm(X_\Gamma \smallsetminus (X_\Gamma \cap
\overline{X_{\Gamma\smallsetminus e}})) \to \fm(\P^{n-1} \smallsetminus X_\Gamma) $$
$$ \to \fm(\P^{n-1} \smallsetminus (X_\Gamma \cap
\overline{X_{\Gamma\smallsetminus e}})) \to \fm(X_\Gamma \smallsetminus
(X_\Gamma \cap \overline{X_{\Gamma\smallsetminus e}}))[1]. $$
We then use the isomorphisms $$ X_\Gamma \smallsetminus (X_\Gamma \cap
\overline{X_{\Gamma\smallsetminus e}}) \simeq \P^{n-2}\smallsetminus X_{\Gamma\smallsetminus e} $$
and
$$ X_\Gamma \cap \overline{X_{\Gamma\smallsetminus e}} \simeq
\overline{X_{\Gamma/e}} \cap \overline{X_{\Gamma\smallsetminus e}} $$
proved in Theorem \ref{delcon} to get the triangle \eqref{delconDM}.
\endproof

This means that we can upgrade at the level of the category
$\cD\cM_\Q$ of mixed motives some of the arguments that we
formulated in the previous sections at the level of classes 
in the Grothendieck ring of varieties. 

\begin{cor}\label{DMmixedTate}
Let $\Gamma_{me}$ denote the graph obtained from a given graph
$\Gamma$ by replacing an edge $e$ by $m$ parallel edges, as in 
\S \ref{multMotSec}. If $\fm_\Gamma$, $\fm_{\Gamma\smallsetminus e}$ and
$\fm_{\Gamma/e}$ belong to the sub-triangulated category $\cD\cM\cT_\Q
\subset \cD\cM_\Q$ of mixed Tate motives, then the motive
$\fm_{\Gamma_{me}}$ also belongs to $\cD\cM\cT_\Q$. 
\end{cor}

\proof It suffices to show that the result holds for
$\Gamma_{2e}$. We look at the case where $e$ is neither 
a bridge nor a looping edge. The other cases can be handled
similarly. One follows the same argument of
Proposition~\ref{doubling}, written in terms of the 
distinguished triangle \eqref{delconDM}, which is here of the form 
$$ 
\fm_\Gamma \to \fm_{\Gamma_{2e}}\to \fm(\P^{n-1}\smallsetminus (
\overline{X_\Gamma} \cap \overline{X_{\Gamma_o}})) \to \fm_\Gamma [1],
$$
where $\Gamma_o$ is the graph obtained by attaching a looping edge at
the vertex $e$ is contracted to in the graph $\Gamma/e$. Since
$\cD\cM\cT_\Q$ is a sub-triangulated category of $\cD\cM_\Q$, to know
that $\fm_{\Gamma_{2e}}$ is (isomorphic to) an object in
$\cD\cM\cT_\Q$ it suffices to know that the remaining two terms of the
distinguished triangle belong to $\cD\cM\cT_\Q$. This requires
expressing $\fm(\P^{n-1}\smallsetminus (
\overline{X_\Gamma} \cap \overline{X_{\Gamma_o}}))$ in terms of
mixed Tate motives. This can be done again as in
Proposition~\ref{doubling}, using again \eqref{Xo} to control
the term $\fm(\P^{n-1}\smallsetminus (
\overline{X_\Gamma} \cap \overline{X_{\Gamma_o}}))$ in terms 
of another distinguished triangle involving $\fm_{\Gamma\smallsetminus e}$ and
$\fm_{\Gamma/e}$. 
\endproof

Similarly, we can lift at this motivic level the statement of
Corollary \ref{MTMchains} and the construction of the ``lemonade
graphs'' of \S \ref{lemonadeSec}. 

\begin{cor}\label{MTMchainsDM}
The motives $\fm_\Gamma$ of graphs $\Gamma$ that are polygon chains 
belong to the subcategory $\cD\cM\cT_\Q$ of mixed Tate motives.
Moreover, if $\Gamma$ is a graph such that $\fm_\Gamma$,
$\fm_{\Gamma\smallsetminus e}$, and $\fm_{\Gamma/e}$ are onjects in the subcategory
$\cD\cM\cT_\Q$ of mixed Tate motives, all the graphs of the form
$\Gamma^\Lambda_m$, obtained as in Proposition \ref{lemonth2}
by attaching a lemon graph to the edge $e$ also have
$\fm_{\Gamma^\Lambda_m}$ in $\cD\cM\cT_\Q$.
\end{cor}

Another possible way of formulation the deletion--contraction
relations of Theorem \ref{delcon} at the level of the triangulated
category of mixed motives, in the form of distinguished triangles,
would be to use the geometric description of the deletion--contraction
relations given in \S \ref{MilnorSec} in terms of the blowup diagram
\eqref{blowupdiag} and used distinguished triangles in $\cD\cM_\Q$
associated to blowups. 

\smallskip

A related question is then to provide a categorification for the
polynomial invariant $C_\Gamma(T)$. This ties up with the question of
what type of deletion--contraction relation this invariant satisfies
by reformulating the question in terms of a possible long exact 
cohomology sequence.

%\newpage

\bigskip

{\bf Acknowledgments.} The results in this article were catalyzed by a
question posed to the first author by Michael Falk during the Jaca
`Lib60ber' conference, in June 2009. We thank him and the organizers
of the conference, and renew our best wishes to Anatoly Libgober.
We thank Friedrich Hirzebruch for pointing out Remark~\ref{Hirzrem}
to us, thereby triggering the thoughts leading to much of the material
in \S\ref{UnivSec}. We also thank Don Zagier for a useful conversation
on divisibility sequences and generating functions.
The second author is partially supported by NSF grants DMS-0651925 and
DMS-0901221. This work was carried out during a stay of the authors at the
Max Planck Institute in July~2009.

\end{document}